\documentclass[a4paper,11pt,onecolumn,accepted=2021-12-29]
               {quantumarticle}
\pdfoutput=1
\usepackage{tikz}
\usepackage{amsmath}
\usepackage{amssymb}  
\usepackage{amsthm}
\usepackage{amsfonts}
\usepackage{graphicx}
\usepackage{bbm}
\usepackage{url}
\usepackage{comment}

\usepackage{braket}
\usepackage{ upgreek }
\usepackage{dsfont}
\usepackage{makecell}
\usepackage{authblk}
\usepackage[numbers]{natbib}
\usepackage[plain]{fancyref} %
\usepackage[ruled]{algorithm}
\usepackage{algpseudocode}

\usepackage{float}

\usepackage{hyperref}

\newtheorem{bigthm}{Theorem}

\theoremstyle{plain}               
\newtheorem{thm}{Theorem}[section]
\newtheorem{fact}{Fact}[section]
\newtheorem{lem}{Lemma}[section]
\newtheorem{cor}{Corollary}[section]
\newtheorem{prop}{Proposition}[section]

\newtheorem{defi}{Definition}[section]

\theoremstyle{remark}

\newcommand*{\fancyrefthmlabelprefix}{thm}
\newcommand*{\fancyreflemlabelprefix}{lem}
\newcommand*{\fancyrefcorlabelprefix}{cor}
\newcommand*{\fancyrefdefilabelprefix}{defi}
\frefformat{plain}{\fancyreflemlabelprefix}{lemma\fancyrefdefaultspacing#1}
\Frefformat{plain}{\fancyreflemlabelprefix}{Lemma\fancyrefdefaultspacing#1}
\frefformat{plain}{\fancyrefthmlabelprefix}{theorem\fancyrefdefaultspacing#1}
\Frefformat{plain}{\fancyrefthmlabelprefix}{Theorem\fancyrefdefaultspacing#1}
\frefformat{plain}{\fancyrefcorlabelprefix}{corollary\fancyrefdefaultspacing#1}
\Frefformat{plain}{\fancyrefcorlabelprefix}{Corollary\fancyrefdefaultspacing#1}
\frefformat{plain}{\fancyrefdefilabelprefix}{definition\fancyrefdefaultspacing#1}
\Frefformat{plain}{\fancyrefdefilabelprefix}{Definition\fancyrefdefaultspacing#1}
\newcommand*{\fancyrefalglabelprefix}{alg}
\newcommand*{\frefalgname}{algorithm}
\newcommand*{\Frefalgname}{Algorithm}
\frefformat{plain}{\fancyrefalglabelprefix}{%
  \frefalgname\fancyrefdefaultspacing#1%
}%
\Frefformat{plain}{\fancyrefalglabelprefix}{%
  \Frefalgname\fancyrefdefaultspacing#1%
}%

\newcommand*{\fancyrefapplabelprefix}{app}
\newcommand*{\frefappname}{appendix}
\newcommand*{\Frefappname}{Appendix}
\frefformat{plain}{\fancyrefapplabelprefix}{%
  \frefappname\fancyrefdefaultspacing#1%
}%
\Frefformat{plain}{\fancyrefapplabelprefix}{%
  \Frefappname\fancyrefdefaultspacing#1%
}%

\definecolor{Green}{HTML}{00AD69}  %

\def\beq{\begin{equation}}
\def\eeq{\end{equation}}
\def\bq{\begin{quote}}
\def\eq{\end{quote}}
\def\ben{\begin{enumerate}}
\def\een{\end{enumerate}}
\def\bit{\begin{itemize}}
\def\eit{\end{itemize}}

\def\lb{\left(}
\def\rb{\right)}

\def\l|{\left|}
\def\r|{\right|}

\newcommand\C{\mathbbm{C}}

\newcommand\R{\mathbbm{R}}

\newcommand{\ketbra}[1]{|#1\rangle\!\langle#1|}
\newcommand{\tr}[1]{\operatorname{tr}\lb#1\rb}

\newcommand{\rl}[2]{S\lb#1\|#2\rb}

\newcommand{\setC}{\mathbb{C}}

\newcommand{\cO}{\mathcal{O}}
\newcommand{\tcO}{\tilde{\mathcal{O}}}
\newcommand{\sign}{\operatorname{sign}}

\begin{document}

\title{Faster quantum and classical SDP approximations for quadratic binary optimization}

\author[1, 2,3]{Fernando G.S L. Brand\~{a}o,}
\author[1,2,4]{Richard Kueng}
\author[5,6]{and Daniel Stilck Fran\c{c}a}

\affil[1]{\small Institute for Quantum Information and Matter, California Institute of Technology, Pasadena,~CA,~USA}

\affil[2] {\small Department of Computing and Mathematical Sciences, California Institute of Technology, Pasadena,~CA,~USA}

\affil[3]{\small AWS Center for Quantum Computing, Pasadena,~CA,~USA}

\affil[4]{\small Institute for Integrated Circuits, Johannes Kepler University Linz, Austria}

\affil[5]{\small QMATH, Department of Mathematical Sciences, University of Copenhagen, Denmark}

\affil[6]{\small Department of Mathematics, Technische Universit\"at M\"unchen, Germany}

\normalsize

\date{}

\maketitle

\begin{abstract}

We give a quantum speedup for solving the canonical semidefinite programming relaxation for binary quadratic optimization. 
This class of relaxations for combinatorial optimization has so far eluded quantum speedups.
Our methods combine ideas from quantum Gibbs sampling and matrix exponent updates. 
A de-quantization of the algorithm also leads to a faster classical solver. 
For generic instances, our quantum solver gives a nearly quadratic speedup over state-of-the-art algorithms. Such instances include approximating the ground state of spin glasses and \textsc{MaxCut} on Erd\"{o}s-R\'enyi graphs.
We also provide an efficient randomized rounding procedure that converts approximately optimal SDP solutions into approximations of the original quadratic optimization problem.
\end{abstract}

\section{Introduction}

Quadratic optimization problems with binary constraints are an important class of optimization problems. Given a (real-valued) symmetric $n \times n$ matrix $A$ the task is to compute
\begin{align}
\textrm{maximize} \quad &  \langle x| A |x \rangle  \quad \textrm{subject to} \quad x \in \left\{ \pm 1 \right\}^n && (\textsc{MaxQP}). \label{eq:maxqp}
\end{align}
This problem arises naturally in many applications across various scientific disciplines, e.g.\   image compression \cite{Leary1983}, latent semantic indexing \cite{Kolda1998}, community detection~\cite{montanari_semidefinite_2016}, correlation clustering \cite{Charikar2004,Mei2017} and structured principal component analysis, see e.g.\ \cite{Kueng2019a,Kueng2019b} and references therein.
Mathematically, \textsc{MaxQP}s \eqref{eq:maxqp} are closely related to computing the $\infty \to 1$ norm of matrices. This norm, in turn, closely relates to the \emph{cut norm} (replace $x \in \left\{\pm 1 \right\}^n$ by $x \in \left\{0,1 \right\}^n$), as both norms can only differ by a constant factor. These norms are an important concept in theoretical computer science \cite{Frieze1999,Alon2003,Alon2006},
since problems such as identifying the largest cut in a graph (\textsc{MaxCut}) can be naturally formulated as instances of these norms.
This connection highlights that optimal solutions of \eqref{eq:maxqp} are \textsc{NP}-hard to compute in the worst case. 
Despite their intrinsic hardness, quadratic optimization problems do admit a canonical \emph{semidefinite programming} (SDP) relaxation\footnote{Rewrite the objective function in \eqref{eq:maxqp} as $\tr{A |x \rangle \! \langle x|}$ and note that every matrix $X = |x \rangle \! \langle x|$ with $x \in \left\{ \pm 1 \right\}^n$ has diagonal entries equal to one and is psd with unit rank. Dropping the (non-convex) rank constraint produces a convex relaxation.}
 \cite{Goemans1995}:
\begin{align}
\textrm{maximize} \quad& \tr{A X} \quad \textrm{subject to} \quad \mathrm{diag}(X)=\mathbf{1},\; X \geq 0 & (\textsc{MaxQP SDP}) \label{eq:maxqpsdpnormal}
\end{align}
Here, $X \geq 0$ indicates that the $n \times n$ matrix $X$ is positive semidefinite (psd), i.e.\ $\langle y|X|y \rangle \geq 0$ for all $y \in \mathbb{R}^n$. 
SDPs comprise a rich class of convex optimization problems that can be solved efficiently under mild assumptions, e.g.\ by using interior-point methods \cite{Boyd2004}. 

Perhaps surprisingly, the optimal value of the \textsc{MaxQP} relaxation often provides a constant factor approximation to the optimal value of the original quadratic problem. However, the associated optimal matrix $X^\sharp$ is typically \emph{not} in one-to-one correspondence with an optimal feasible point $x^\sharp \in \left\{ \pm 1 \right\}^n$ of the original problem \eqref{eq:maxqp}.
Several randomized rounding procedures have been devised to overcome this drawback since the pioneering work of~\cite{Goemans1995}. These transform $X^\sharp$ into a random binary vector $\tilde{x} \in \left\{ \pm 1 \right\}^n$ that achieves $\langle \tilde{x}|A | \tilde{x} \rangle \geq \gamma \max_{x \in \left\{ \pm 1 \right\}^n}\langle x|A|x \rangle$ in expectation for some constant $\gamma$. Explicit  values of $\gamma$ are known for instance for the case of $A$ being the adjacency matrix of a graph~\cite{Goemans1995} or positive semidefinite~\cite{Alon2006}.

Although tractable in a theoretical sense, the runtime associated with general-purpose SDP solvers quickly becomes prohibitively expensive in both memory and time. This practical bottleneck has spurred considerable attention in the theoretical computer science community over the past decades \cite{Arora2005,Burer2005,Boumal2016,Tropp2017}.
(Meta) algorithms, like \emph{matrix multiplicative weights} (MMW) \cite{Arora2005} solve the \textsc{MaxQP SDP} \eqref{eq:maxqpsdpnormal} up to multiplicative error $\epsilon\|A\|_{\ell_1}$ in runtime $\cO ((n/\epsilon)^{2.5} s )$, where $s$ denotes the column sparsity of $A$. 
Further improvements are possible if the problem description $A$ has additional structure, such as $A$ being the adjacency matrix of a graph~\cite{Arora2007}.

Very recently, a line of works pointed out that quantum computers can solve certain SDPs even faster \cite{Brandao2017a,Apeldoorn2017, Apeldoorn2018,Brandao2017b,Kerenidis2018}. From a high-level perspective, the existing approaches fall into two categories: those that obtain quantum speedups by capitalizing on the fact that quantum computers can prepare certain quantum Gibbs states faster~\cite{Brandao2017a,Apeldoorn2017, Apeldoorn2018,Brandao2017b} and those that capitalize on the fact that quantum computers can solve certain linear algebra problems faster~\cite{Kerenidis2018}. For the works that focus on quantum Gibbs states, we further distinguish those that follow a primal-dual approach~\cite{Brandao2017a,Apeldoorn2017, Apeldoorn2018,Brandao2017b} and those that follow a primal-only approach~\cite{Apeldoorn2018}\footnote{Note that \cite{Apeldoorn2018} proposes both a primal-only and a primal-dual algorithm.}.
For both the linear algebra and the primal-dual approach the current runtime guarantees depend on problem-specific parameters. These parameters scale particularly poorly for most combinatorial optimization problems, including the \textsc{MaxQP SDP}, and negate any potential advantage. Moreover, other closely related primal only approaches~\cite{Apeldoorn2018} do not yield a quantum speedup straight away because they treat the constraints in a black-box manner.

In this work, we tackle this challenge and overcome the shortcomings of existing quantum SDP solvers by considering the following further relaxation of problem \eqref{eq:maxqpsdpnormal}:
\begin{align}\label{equ:maxqpsdp}
\textrm{find} &\quad X \quad \textrm{(renormalized, relaxed, feasibility \textsc{MaxQP SDP})}\\
\textrm{subject to} &\quad  \mathrm{tr} \left( \tfrac{1}{\|A\|}AX \right)\geq\lambda-\epsilon \nonumber\\
& \quad \sum_i\left|\langle i| X |i \rangle-\frac{1}{n}\right|\leq \epsilon\nonumber\\
& \quad \mathrm{tr}(X)=1,\;   X \geq 0 \nonumber.
\end{align}

Here we introduced two additional parameters $\lambda$ and $\epsilon$. The $\lambda$ parameter comes from a standard trick to reduce the problem in~\eqref{eq:maxqpsdpnormal} to a sequence of feasibility problems, as we will explain later. The $\epsilon$ parameter encodes a further relaxation of the constraints of Eq.~\eqref{eq:maxqpsdpnormal}. Let us first discuss the case where $\epsilon=0$, that is, we have the normalized diagonal constraints $\langle i| X |i \rangle = \tfrac{1}{n}$.

This renormalization of the original problem pinpoints connections to quantum mechanics: Every feasible point $X$ obeys $\tr{X}=1$ and $X \geq 0$, implying that it describes the state $\rho$ of a $n$-dimensional quantum system. 
In turn, such quantum states can be represented approximately by a renormalized matrix exponential $\rho = \exp (-H)/\mathrm{tr}( \exp (-H))$, the \emph{Gibbs state} associated with \emph{Hamiltonian} $H$. We capitalize on this correspondence 
by devising a meta-algorithm -- \emph{Hamiltonian Updates} (HU) -- that is inspired by matrix exponentiated gradient updates \cite{Tsuda2005}, see also \cite{Steurer2015, Brandao2017b,Hazan2006} for  similar approaches.
Another key insight is that the diagonal constraints also have a clear quantum mechanical interpretation: the feasible states are those that are indistinguishable from the uniform or maximally mixed state when measured in the computational basis.  

This interpretation points the way to another key component to obtaining speedups for \textsc{MaxQP SDP}: by setting $\epsilon>0$ we further relax the problem and optimize over all states that are \emph{approximately} indistinguishable from the maximally mixed state when measured in the computational basis. This further relaxation will allow us to overcome shortcomings of previous solvers when dealing with SDPs of this form, as it bundles up the $n$ linear constraints in Eq.~\eqref{equ:maxqpsdp} into one. 
As we ultimately want to solve Eq.~\eqref{eq:maxqpsdpnormal} and not ~\eqref{equ:maxqpsdp}, a significant part of the technical contribution of this work is to show that this further relaxation is mild. Indeed, we will be able to convert a solution to~\eqref{equ:maxqpsdp} into one to~\eqref{eq:maxqpsdpnormal} by only slightly changing the objective value.

This will allow us to solve the relaxed problem up to an $\epsilon$ additive error by only imposing a relaxation parameter $\epsilon^{4}$. As is the case with this and many other related algorithms to solve SDPs, ensuring that we only require a dimension-independent $\epsilon^{4}$ precision for the constraints is essential to guarantee speedups.
Note that to obtain the same level of precision as in the formulation given in~\eqref{equ:maxqpsdp} would require enforcing that each diagonal constraint is satisfied up to an error of order $\cO(n^{-1})$.

Although originally designed to exploit the fact that quantum architectures can sometimes create Gibbs states efficiently and inspired by interpreting the problem from the point of view of quantum mechanics, it turns out that this approach also produces faster classical algorithms.

To state our results, we instantiate standard computer science notation. The symbol $\cO (\cdot)$ describes limiting function behavior, while $\tcO (\cdot)$ hides poly-logarithmic factors in the problem dimension and polynomial dependencies on the inverse accuracy $1/\epsilon$. We are working with the adjacency list oracle model, where individual entries and location of nonzero entries of the problem description $A$ can be queried at unit cost. We refer to Section~\ref{sec:quantum} for a more detailed discussion.

\begin{bigthm}[Hamiltonian Updates: runtime] \label{thm:main_intro}
Let $A$ be a (real-valued), symmetric $n \times n$ matrix with column sparsity $s$.
Then, the associated \textsc{MaxQP SDP} \eqref{eq:maxqpsdpnormal} can be solved up to additive accuracy $n\|A\|\epsilon$  in runtime $\tcO\lb n^{1.5}\lb \sqrt{s}\rb^{1+o(1)}\epsilon^{-28+o(1)}\operatorname{exp}(1.6\sqrt{12\log(\epsilon^{-1})})\rb$ on a quantum computer and $\tcO\left( \min \{n^2s, n^\omega \}\epsilon^{-12}\right)$ on a classical computer.
\end{bigthm}

Here $\omega$ is the matrix multiplication exponent. With some abuse of terminology, the word ``solves" is used with slightly different meanings for the classical and quantum algorithms in the statement above. For the classical algorithm we can indeed output a feasible solution of  \textsc{MaxQP SDP} that is $n\epsilon$ close to the optimal target value. In the quantum case, the output is in the form of a quantum state $\rho$ such that $n\rho$ is $\cO(n\epsilon)$ close in trace distance to a feasible point and with value that is $n\|A\|\epsilon$ close, what we will call approximately feasible.
We emphasize that the quantum algorithm also outputs a classical description of a solution that is approximately feasible in a sense that will be made precise below. 
The polynomial dependency on inverse accuracy is rather high (e.g.\ $(1/\epsilon)^{12}$ for the classical algorithm). We expect future work to be able to
improve this.

Already the classical runtime improves upon the best known existing results and we refer to Section~\ref{sub:related} for a detailed comparison. Access to a quantum computer would increase this gap further.
However, it is important to point out that Theorem~\ref{thm:main_intro} has an approximation error of order $n \| A \| \epsilon$. 
In contrast, MMW \cite{Arora2005} -- the fastest existing algorithm --  incurs an error proportional to $\epsilon\| A \|_{\ell_1}$, where $\| A \|_{\ell_1}=\sum_{i,j}|A_{i,j}|$, making a straightforward comparison more difficult. 
Importantly, the scaling of our algorithm is favorable for generic problem instances and spin glass models, see Section~\ref{sub:related}.

The quantum algorithm outputs a classical description of a Hamiltonian $H^\sharp$ that encodes an approximately optimal, approximately feasible solution $\rho^\sharp = \exp (-H^\sharp)/\tr{\exp (-H^\sharp)}$ of the renormalized \textsc{MaxQP SDP} \eqref{equ:maxqpsdp}. 
This classical output can subsequently be used for randomized rounding for the $\infty \to 1$ norm of a matrix $A$, $\|A\|_{\infty\to1}=\max\limits_{x,y\in \{\pm1\}^n}\langle x|A|y\rangle$.

\begin{bigthm}[Rounding]
Suppose that $H^\sharp$ encodes an approximately optimal solution of the renormalized \textsc{MaxQP SDP} \eqref{equ:maxqpsdp} with accuracy $\epsilon^4$  for the target matrix
\begin{align*}
    A'=\left( \begin{array}{cc}
0& A \\
A^T &0
\end{array}
\right),
\end{align*}
where $A$ is a $n\times n$ real matrix with at most $s$ nonzero entries per column (column sparsity).
Then, there is a classical $\tcO(ns)$-time randomized rounding procedure that converts $H^\sharp$ 
 into binary vectors $\tilde{x},\tilde{y} \in \left\{ \pm 1\right\}^n$ that obey
\begin{equation*}
\gamma\left( \|A\|_{\infty\to1} - \cO (n \| A \|\epsilon) \right) \leq \mathbb{E} \left[ \langle \tilde{x} | A |\tilde{y} \rangle \right] \leq \|A\|_{\infty\to1}  \label{eq:rounding_main},
\end{equation*}
where $\gamma=\frac{2}{\pi}$ if $A$ is positive semidefinite and $\frac{4}{\pi}-1$ else.
\end{bigthm}

This result recovers the randomized rounding guarantees of \cite{Alon2006} in the limit of perfect accuracy ($\epsilon=0$). 
However, for $\epsilon >0$ the error scales with $n \| A \|$. 
In turn, randomized rounding only provides a multiplicative approximation if $ \|A\|_{\infty\to1}$ is of the same order.  This result on the randomized rounding also relies on a detailed analysis of the stability of the rounding procedure w.r.t. to approximate solutions to the problem.

\section{Detailed summary of results}\label{sec:algofeas}

We present \emph{Hamiltonian Updates} --  a meta-algorithm for solving convex optimization problems over the set of quantum states based on quantum Gibbs sampling -- in a more general setting, as we expect it to find applications to other problems.
Throughout this work, $\| \cdot \|_{tr}$ and $\| \cdot \|$ denote the trace (Schatten-1) and operator (Schatten-$\infty$) norms, respectively.

\subsection{Convex optimization and feasibility problems}

SDPs over the set  of quantum states are a special instance of a more general class of convex optimization problems. For a bounded, concave function $f$ from the set of symmetric matrices to the real numbers and closed convex sets $\mathcal{C}_1,\ldots,\mathcal{C}_n$, solve
\begin{align}
\textrm{maximize} & \quad f(X)& (\textsc{CPopt}) \label{eq:CPopt} \\
\textrm{subject to} & \quad X \in \mathcal{C}_1 \cap \cdots \cap \mathcal{C}_m,& \nonumber\\
& \quad \mathrm{tr}(X)=1,\;  X \geq 0.& \nonumber
\end{align}
The constraint $\mathrm{tr}(X)=1$ enforces normalization, while $X \geq 0$ is the defining structure constraint of semidefinite programming. Together, they restrict $X$ to the set of $n$-dimensional quantum states $\mathcal{S}_n = \left\{ X:\; \mathrm{tr}(X)=1,\; X \geq 0 \right\}$. 
We will now specialize to the case $f(A)=\tr{AX}$ for a symmetric matrix $A$, as this is our main case of interest, but remark that it is simple to generalize the discussion that follows for more general classes. 
This trace normalization constraint implies fundamental bounds on the optimal value:
 $| \mathrm{tr}(A X^\sharp)| \leq \| A \| \| X^\sharp \|_{tr}=\| A \|$, according to Matrix H\"older \cite[Ex.~IV.2.12]{Bhatia1997}.
Binary search over potential optimal values $\lambda \in \left[-\| A \|, \| A \| \right]$ 
allows for reducing the convex optimization problem into a sequence of feasibility problems:
\begin{align}
\textrm{find} & \quad X \in \mathcal{S}_n & (\textsc{CPfeas}(\lambda)) \label{eq:CPfeas}\\
\textrm{subject to} & \quad \tr{A X} \geq \lambda, \nonumber&\\
& \quad X \in \mathcal{C}_1 \cap \cdots \cap \mathcal{C}_m. \nonumber&
\end{align}
The convergence of binary search is exponential.
This ensures that the overhead is benign: a total of $\log (\| A \|/\epsilon)$ queries of $\textsc{CPfeas}(\lambda)$ suffices to determine the optimal solution of \textsc{CPopt} \eqref{eq:CPopt} up to accuracy $\epsilon$. In summary: 

\begin{fact} \label{fact:binsearch}
Binary search reduces the task of solving convex optimization problems \eqref{eq:CPopt} to the task of solving convex feasibility problems \eqref{eq:CPfeas}. 
\end{fact}

\subsection{Meta-algorithm for approximately solving convex feasibility problems}

We adapt a meta-algorithm developed by Tsuda, R\"atsch and Warmuth \cite{Tsuda2005}, see also \cite{Steurer2015,Arora2007,Hazan2006,Brandao2017b} for similar ideas and~\cite{Bubeck2015} for an overview of these techniques. All these algorithms, including the variation presented here, can be seen as instances of mirror descent with the mirror map given by the von Neumann entropy with adaptations tailored to the problem at hand. We believe our variation provides a path for also obtaining quantum speedups for nonlinear convex optimizations, so we state it in more detail.

For our algorithm, we require subroutines that allow for testing $\epsilon$-closeness (in trace norm) to each convex set $\mathcal{C}_i$. 

\begin{defi}[$\epsilon$-separation oracle] \label{def:oracle}
Let $\mathcal{C}\subset\mathcal{S}_n$ be a closed, convex subset of quantum states and $\mathcal{C}^*\subset\{X=X^\dagger\in\C^{n\times n}: \|X\|\leq 1\} $ be a closed, convex subset of observables of operator norm at most $1$.
For $\epsilon>0$ an $\epsilon$-separation oracle (with respect to $\mathcal{C}^*$) is a subroutine that 
either accepts a state $\rho$ (in the sense that observables from $\mathcal{C}^*$ cannot distinguish $\rho$ from elements of $\mathcal{C}$), or provides a hyperplane $P$ that separates $\rho$ from the convex set using a test from $\mathcal{C}^*$:
\begin{equation*}
O_{\mathcal{C}, \epsilon}(\rho)
= \begin{cases}
\textrm{accept } \rho \textrm{ if }\min_{Y \in \mathcal{C}} \max_{P\in \mathcal{C}^*}\mathrm{tr}(P (\rho-Y)) \leq  \epsilon, \\
\textrm{else: output } P\in \mathcal{C}^* \textrm{s.t.}\; \mathrm{tr}(P (\rho-Y)) \geq \frac{\epsilon}{2} \; \textrm{for all}\; Y \in \mathcal{C}.
\end{cases}
\end{equation*}
\end{defi}
We note that the Oracle is well-defined in the sense that if $\min_{Y \in \mathcal{C}} \max_{P\in \mathcal{C}^*}\mathrm{tr}(P (\rho-Y)) > \epsilon$, then there exists $P$ such that $P\in \mathcal{C}^*$ and for all $Y\in \mathcal{C}$
\begin{align*}
    \mathrm{tr}(P (\rho-Y)) \geq \frac{\epsilon}{2}.
\end{align*}
Indeed, by Sion's min-max theorem~\cite{sion_general_1958} we have
\begin{align*}
    &\max\limits_{P\in \mathcal{C}^*}\min\limits_{Y\in\mathcal{C}}\tr{P(\rho-Y)}= \min\limits_{Y\in\mathcal{C}}\max\limits_{P\in \mathcal{C}^*}\tr{P(\rho-Y)}>\epsilon.
\end{align*}
This implies that there even exists a $P$ that separates the state $\rho$ from the set $\mathcal{C}$ by $\epsilon$.
Nonetheless, we instantiate the weaker requirement with only $\epsilon/2$ separation. This will be vital to ensure that the algorithm can tolerate errors and/or approximations in the samples from $\rho$.

By allowing for fine-tuning of $\mathcal{C}^*$ we are able to reduce the number of closeness conditions we need to test.
\emph{Hamiltonian Updates} (\textsc{HU}) is a general meta-algorithm for approximately solving convex feasibility problems~\eqref{eq:CPfeas} (\textsc{CPfeas}). The task is to find a state $\rho$ that is $\epsilon$-close to each convex set $\mathcal{C}_i$ with respect to observables in some $\mathcal{C}^*_i$ ($\max_{P_i\in\mathcal{C}_i^*}\min_{Y_i \in \mathcal{C}_i} \tr{P_i(\rho-Y_i)} \leq \epsilon$) and also obeys %
$\rho \in \mathcal{S}_n$ ($\rho \geq 0$ and $\mathrm{tr}(\rho)=1$).
A change of variables takes care 
of positive semidefiniteness and normalization: replace $X$ in problem~\eqref{eq:CPfeas} by a Gibbs state $ \rho_H= \exp \left( - H \right)/\mathrm{tr}(\mathrm{exp}(-H))$. At each iteration, we query $\epsilon$-separation oracles.
If they all accept, the current iterate is $\epsilon$-close to feasible in the sense that there is a matrix in each $C_i$ whose expectation value with respecto to observables in $C_i^*$ is $\epsilon$ close with respect to the accepted state, and we are done. Otherwise, we update the matrix exponent to penalize infeasible directions: $H \to H+\tfrac{\epsilon}{16} P$, where $P$ is a separating hyperplane that witnesses infeasibility. 
This process is visualized in Figure~\ref{fig:HU} and we refer to Algorithm~\ref{alg:HU} for a detailed description.

\begin{algorithm}[H]
\caption{\textit{Meta-Algorithm for approximately solving convex feasibility problems \eqref{eq:CPfeas}.}
}
\label{alg:HU}
\begin{algorithmic}[1]
\Require{Query access to $m$ $\epsilon$-separation oracles $O_{1,\epsilon}(\cdot),\ldots,O_{m,\epsilon}(\cdot)$}
\Function{HamiltonianUpdates}{$T,\epsilon$}
\State{$\rho=n^{-1} I$ and $H = 0$}
\Comment initialize the maximally mixed state
\For{$t=1,\ldots,T$}
\For{$i=1,\ldots,m$} \Comment{Query oracles and check feasibility}
\If{$O_{i,\epsilon}(\rho)=P$}
\State{$H \leftarrow H + \tfrac{\epsilon}{16}P$} \Comment{Penalize infeasible direction}
\State{$\rho \leftarrow \exp \left( - H \right) / \mathrm{tr}(\exp (-H) )$} \Comment{Update quantum state}
\State{\textbf{break loop}}
\EndIf
\EndFor
\State{return $(\rho,H)$ and \textbf{exit function}} \Comment{Current iterate is $\epsilon$-feasible}
\EndFor
\EndFunction
\end{algorithmic}
\end{algorithm}

\begin{thm}[\textsc{HU}: convergence] \label{thm:main}
Algorithm~\ref{alg:HU} requires at most $T = \lceil 64 \log (n) /\epsilon^2 \rceil+1$ iterations to either certify that \eqref{eq:CPfeas} is infeasible or output a state $\rho$ satisfying: 
\begin{align}\label{equ:whatfeasibility}
\textrm{for all } 1\leq i\leq m: \max_{P_i\in\mathcal{C}_i^*}\min_{Y_i \in \mathcal{C}_i} \tr{P_i(\rho-Y_i)} \leq \epsilon
\end{align}
\end{thm}
As it is also the case for the aforementioned variations of the algorithm above, the proof follows from establishing sufficiently large step-wise progress in quantum relative entropy. The quantum relative entropy between \emph{any} feasible state and the initial state $\rho_0=n^{-1}I$ (maximally mixed state) is bounded by $\log (n)$. Therefore, the algorithm must terminate after sufficiently many iterations. Otherwise, the problem is infeasible. We refer to Section ~\ref{sec:proof_main} for details. Note that, unlike related previous quantum solvers~\cite{Brandao2017b,Apeldoorn2017} our algorithm only considers the primal problem and bears a strong resemblance with the primal-only approach of~\cite{Apeldoorn2018}, although we differ in the way we implement the constraints.

Theorem~\ref{thm:main} has important consequences: The runtime of approximately solving quantum feasibility problems is dominated by the cost of implementing $m$ separation oracles $O_{i,\epsilon}$ and the cost associated with matrix exponentiation. This reduces the task of efficiently solving convex feasibility problems to the quest of efficiently identifying separating hyperplanes and developing fast routines for computing Gibbs states. 

The latter point already hints at a genuine quantum advantage: quantum architectures can efficiently prepare (certain) Gibbs states \cite{Chowdhury2016,Franc2017,Kastoryano2014,Poulin2009,Temme2009,Temme2009,Yung2012,Apeldoorn2017}.

It should be stressed that the approximate feasibility guarantee in~\eqref{equ:whatfeasibility} is not very strong and a careful choice of the $C_i,C_i^*$ and a careful analysis of the continuity of the problem is usually required to ensure that it gives a good approximation to \textsc{CPopt} \eqref{eq:CPopt}.

\subsection{Classical and quantum solvers for the renormalized \textsc{MaxQP SDP}}
Let us now formulate the renormalized, relaxed problem in Eq.~\eqref{equ:maxqpsdp} in this framework and discuss the appropriate oracles.
For fixed $\lambda \in [-1,1]$ the (feasibility) \textsc{MaxQP SDP} is equivalent to a quantum feasibility problem:
\begin{align*}
\textrm{find} &\quad \rho \in \mathcal{S}_n \cap \mathcal{A}_\lambda \cap \mathcal{D}_n \\
\textrm{where} & \quad \mathcal{A}_\lambda = \left\{ X:\; \mathrm{tr} \left( A \| A \|^{-1} X \right) \geq \lambda \right\},\quad \mathcal{A}_\lambda^* = \left\{ -A \|A\|^{-1}\right\} \\
& \quad \mathcal{D}_n = \left\{ X:\; \langle i| X |i \rangle =1/n, i \in [n]\right\},\quad \mathcal{D}_n^* = \left\{ X:\; \|X\|\leq1, X \textrm{ is diagonal}\right\}.
\end{align*}
The set $\mathcal{A}_\lambda$ corresponds to a half-space, while $\mathcal{D}_n$ is an affine subspace with codimension $n$. Let us now see that the convergence promises of Thm.~\ref{thm:main} indeed convert to the renormalized, relaxed, feasibility \textsc{MaxQP SDP}, see Eq.~\eqref{equ:maxqpsdp}. Let us start with observing
\begin{align}\label{equ:whatfeasibilityA}
\max\limits_{P\in \mathcal{A}_\lambda^*}\min\limits_{Y\in\mathcal{A}_\lambda}\tr{P(\rho-Y)}\leq \epsilon\Longleftrightarrow -\mathrm{tr} \left( A \| A \|^{-1} (\rho-Y) \right) \leq \epsilon\quad \text{for all $ Y\in\mathcal{A}_\lambda$.}
\end{align}
Combined with the defining halfspace condition for $\mathcal{A}_\lambda$, this display asserts $\mathrm{tr} \left( A \| A \|^{-1} \rho\right) \geq \lambda-\epsilon$.
We can analyze the oracle for $\mathcal{D}_n$ in a similar fashion. Note that,
\begin{align}\label{equ:whatfeasibility_diagonal}
\max\limits_{P\in \mathcal{D}_n^*}\min\limits_{Y\in\mathcal{D}_n}\tr{P(\rho-Y)}\leq \epsilon\Longleftrightarrow\sum\limits_{i=0}^{n-1}\left|\langle i| \rho|i \rangle-1/n\right| \leq\epsilon.
\end{align}
Thus, we indeed obtain Eq.~\eqref{equ:maxqpsdp} from this formulation up to an error $\epsilon$ for the target value.  It will be important to ensure that both quantum and classical algorithms work only having access to approximations of the current iteration.
The simple structure of both sets readily suggests two separation oracles that take this into account:
\begin{enumerate}
\item[$\mathcal{O}_{\mathcal{A}_\lambda}$:] compute an approximation $\tilde{a}\in\R$ up to additive error $\frac{\epsilon}{4}$ of $\mathrm{tr}(A \|A \|^{-1} \rho)$. Check if $\tilde{a}\geq \lambda-\frac{3\epsilon}{4}$ and output $P= -A \| A \|^{-1}$ if this is not the case.
\item[$\mathcal{O}_{\mathcal{D}_n}$:] compute an approximation $\tilde{p}\in\R^n$ of $p(i)=\langle i| \rho |i \rangle$ satisfying $\sum_i|p(i)-\tilde{p}(i)|\leq\frac{\epsilon}{4}$. Check $\sum_i|\tilde{p}(i) -1/n|\leq\frac{3\epsilon}{4}$  and output 
\begin{align}\label{equ:definition_P}
   P= \sum_{i=1}^n (\mathbb{I} \left\{ \tilde{p}(i) >1/n \right\}-\mathbb{I} \left\{ \tilde{p}(i) <1/n \right\}) |i \rangle \! \langle i|  
\end{align}
if this is not the case.
\end{enumerate}
Note that the oracles are only defined for quantum states as inputs. Let us briefly check that it satisfies the definitions in~\ref{def:oracle}. For $\mathcal{O}_{\mathcal{A}_\lambda}$ we have that if $\tilde{a}\geq \lambda-\frac{3\epsilon}{4}$, then $\mathrm{tr}(A \|A \|^{-1} \rho)\geq \lambda-\epsilon$, as desired. The other case is similar.

For the oracle for $\mathcal{O}_{\mathcal{D}_n}$, let us first assume that we are in the case that $\sum_i|\tilde{p}(i) -1/n|\geq\frac{3\epsilon}{4}$. Clearly, we have that $P$ defined in Eq.~\eqref{equ:definition_P} is diagonal and of operator norm at most $1$. For ease of notation let $f:[n]\to\{-1,1\}$ be $-1$ if $\tilde{p}(i) <1/n$ and $1$ else. By construction we have for any $Y\in\mathcal{O}_{\mathcal{D}_n}\cap \mathcal{S}_n$:
\begin{align*}
&\tr{P(\rho-Y)}=\sum\limits_{i}f(i)\left(p(i)-\tfrac{1}{n}\right)
\\&\stackrel{(1)}{\geq} -\sum_i|p(i)-\tilde{p}(i)|+\sum\limits_{i}f(i)\left(\tilde{p}(i)-\tfrac{1}{n}\right)=\sum_i|p(i)-\tilde{p}(i)|+\sum_i\left|\tilde{p}(i)-\tfrac{1}{n}\right|\\
&\geq-\frac{\epsilon}{4}+\sum_i|\tilde{p}(i)-\tfrac{1}{n}|\geq\frac{\epsilon}{2},
\end{align*}
where in (1) we used H\"older's inequality.
On the other hand, if $\sum_i|\tilde{p}(i) -1/n|\leq\frac{3\epsilon}{4}$, then an application of the triangle inequality shows that $\sum_i|p(i)-1/n|\leq \epsilon$. Thus, we conclude that both oracles are correct.

The key insight to later obtain quantum speedups for the \textsc{MaxQP SDP} is that the second oracle can be interpreted as trying to distinguish the current state from the maximally mixed through computational basis measurements. This view is similar in spirit to~\cite[Lemma 4.6]{Steurer2015}, although here we focus on using this approach to construct solutions and to show that this notion of approximate feasibility is good enough for the \textsc{MaxQP SDP}.

\subsubsection{Classical runtime}

For fixed $\rho_H = \exp (-H)/\mathrm{tr}\left( \exp (-H) \right)$ both separation oracles are easy to implement on a classical computer given access to $\rho_H$. Hence, matrix exponentiation is the only remaining bottleneck. 
This can be mitigated by truncating the Taylor series for $\exp (-H)$ after $l' = \mathcal{O}(\|H\|+ 1/\epsilon)$ many steps. 
Approximating $\rho$ in this fashion only requires
\begin{align*}
\mathcal{O}(\min \left\{ n^2 s, n^\omega \right\} \log (n) \epsilon^{-1})
\end{align*} 
steps and only incurs an error of $\epsilon$ in trace distance. Moreover, it is then possible to convert an approximately feasible to point to a strictly feasible one with a similar value, see Section~\ref{sec:classical}.
The following result becomes an immediate consequence of Fact~\ref{fact:binsearch} and Theorem~\ref{thm:main}.

\begin{cor}[Classical runtime for the \textsc{MaxQP SDP}] \label{cor:classical}
Suppose that $A$ has row-sparsity $s$. Then, the 
classical cost of solving the associated (renormalized) \textsc{MaxQP SDP} up to additive error $\epsilon$ is $\cO( \min \{n^2s, n^\omega \} \log (n)\epsilon^{-12})$.
\end{cor}

The comparatively poor accuracy scaling with $\epsilon^{-12}$ stems largely from the fact that we need to convert an approximately feasible optimal solution into a strictly feasible optimal solution. This conversion is contingent on running Algorithm~\ref{alg:HU} with accuracy $\tilde{\epsilon}=\epsilon^4 \ll \epsilon$ (see Proposition~\ref{prop:stability} below). The total accuracy scaling $\epsilon^{-12} = \tilde{\epsilon}^{-3}$ results from combining the $\cO (\log (n)/\tilde{\epsilon})$-cost for approximating the matrix exponential within a single iteration with the iteration bound $T=O(\log (n)/\tilde{\epsilon}^2)$ from Theorem~\ref{thm:algocorrect}.

\subsubsection{Quantum runtime}

Quantum architectures can efficiently prepare (certain) Gibbs states and are therefore well suited to overcome the main classical bottleneck. In contrast, checking feasibility becomes more challenging, because information about $\rho$ is not accessible directly. Instead, we must prepare multiple copies of $\rho$ and perform quantum mechanical measurements to test feasibility:
\begin{itemize}
\item $\cO(\epsilon^{-2})$ copies of $\rho$ suffice to $\epsilon$-approximate $\mathrm{tr}(A \|A \|^{-1} \rho)$ via phase estimation.
\item $\cO(n\epsilon^{-2})$ copies suffice to estimate the diagonal entries of $\rho$ (up to accuracy $\epsilon$ in trace norm) with a high probability of success via repeated computational basis measurements.
\end{itemize}
Combining this with the overall cost of preparing a single Gibbs state implies the following runtime for executing Algorithm~\ref{alg:HU} on a quantum computer. This result is based on the \emph{sparse oracle input model} and we refer to Sec.~\ref{sec:quantum} for details.

\begin{cor}[Quantum runtime for the \textsc{MaxQP SDP}] \label{cor:quantum}
Suppose that $A$ has row-sparsity $s$. 
Then, the quantum cost of solving the \textsc{MaxQP SDP}  up to additive error
$\epsilon n\|A\|$ is
\begin{align*}
    \tcO(n^{1.5} s^{0.5+o(1)}\mathrm{poly}(1/\epsilon)).
\end{align*}
\end{cor}

The quantum algorithm also outputs a classical description of the Hamiltonian $H^\sharp$ corresponding to an approximately optimal, approximately feasible Gibbs state  and its value. More precisely, it outputs a real number $a$ and a diagonal matrix $D$ such that $H^\sharp=aA+D$ and $n\rho_{H^\sharp}$ is $\cO(n\epsilon)$ close in trace distance to a feasible point of \textsc{MaxQP SDP}.
Moreover, we have the potential to produce samples from the associated approximately optimal Gibbs state $\rho^\sharp = \exp (-H^\sharp)/\tr{\exp (-H^\sharp)}$ in sub-linear runtime $\tcO (\sqrt{n})$  on a quantum computer. 
In the next section we show that the output of the algorithm is enough to give rise to good randomized roundings.

\subsection{Randomized rounding} \label{sec:results_rounding}

The renormalized \textsc{MaxQP SDP} \eqref{equ:maxqpsdp} arises as a convex relaxation of an important quadratic optimization problem \eqref{eq:maxqp}. However, the optimal solution $X^\sharp$ is typically not of the form $|x \rangle \! \langle x|$, with $x \in \left\{ \pm 1 \right\}^n$. 
Goemans and Williamson \cite{Goemans1995} pioneered randomized rounding techniques that allow for converting $X^\sharp$ into a cut $x^\sharp$ that is close to optimal. However, their rounding techniques rely on the underlying matrix being entrywise positive and a more delicate analysis is required to derive analogous results for broader classes of matrices. We will now follow the analysis of~\citep{Alon2006} to do the randomized rounding for the $\infty\to 1$ norm. First, let us make the connection between this norm and the \textsc{MaxQP SDP} clearer. Let $A$ be a real matrix and define 
\begin{align*}
    A'=\left( \begin{array}{cc}
0& A \\
A^T &0
\end{array}
\right).
\end{align*}
It is easy to see that for two binary vectors $x,y\in \{\pm1\}^n$ we have $\langle x\oplus y|A'| x\oplus y\rangle=2\langle x|A|y\rangle$
(with a slight abuse of notation, we also use the bra-ket notation for inner products of unnormalized vectors). This immediately shows that $2\|A\|_{1\to\infty}=\max_{z\in \{\pm 1\}^{2n}}\langle z|A'|z\rangle$, which is an instance of \textsc{MaxQP}.
We will now show that the rounding procedure is stable, i.e.\ randomized rounding of an approximately feasible, approximately optimal point, such as the ones outputted by the quantum algorithm, still results in a good binary vector for approximating this norm. We strengthen the stability of the rounding even further by showing that rounding with a truncated Taylor expansion of the solution is still good enough, saving runtime. The rounding procedure is described in Algorithm~\ref{alg:rounding}.

\begin{algorithm}[H]
\begin{algorithmic}[1] 
\Function{RandomizedRounding}{$H^\sharp,\epsilon$}
\State Draw a random vector $g\in\R^n$ with i.i.d. $\mathcal{N}(0,1)$ entries.
\State Compute $z=\sum_{k=0}^l\frac{(-H^\sharp)^k}{2^kk!}g$ for $l=\cO (\|H^\sharp \|+\log (1/\epsilon))$.%
\State \textbf{output} $x_i=\sign(z_i)$.
\EndFunction
\end{algorithmic}
\caption{\textit{Randomized rounding based on optimal Hamiltonian $H^\sharp$}}\label{alg:rounding}
\end{algorithm}

\begin{prop} \label{prop:rounding_main}
Let $A$ be a real matrix and $H^\sharp$ be such that $\rho^\sharp = \exp (-H^\sharp)/\mathrm{tr}(\exp (-H^\sharp))$ is an $\epsilon$-approximate solution to the renormalized \textsc{MaxQPSDP} for $A'$ \eqref{equ:maxqpsdp} with value $\alpha^\sharp = \mathrm{tr} \left( A' \| A' \|^{-1} \rho^\sharp \right)$. 
Then, the (random) output $x=(x_1\oplus x_2) \in \left\{ \pm 1 \right\}^{2n}$ of Algorithm~\ref{alg:rounding} can be computed in $\tcO (ns)$-time and obeys
\begin{equation*}
\gamma n \| A \| \left(\alpha^\sharp - \mathcal{O}(\epsilon)\right) \leq \mathbb{E} \langle x_1| A |x_2 \rangle \leq n \| A \| (\alpha^\sharp+\mathcal{O}(\epsilon)),
\end{equation*}
where $\gamma=2/\pi$ for $A$ p.s.d. and $4/\pi-1$ else.
\end{prop}

This rounding procedure is fully classical and can be executed in runtime $\tcO(ns)$. We refer to Sec.~\ref{sec:rounding} for details. 
What is more, it applies to both quantum and classical solutions of the \textsc{MaxQP SDP}. Even the quantum algorithm provides $H^\sharp$ in classical form, while the associated $\rho^\sharp$ is only available as a quantum state. 
Rounding directly with $\rho^\sharp$ would necessitate a fully quantum rounding technique that, while difficult to implement and analyze, seems to offer no advantages over the classical Algorithm~\ref{alg:rounding}. Thus, it is possible to perform the rounding even with the output of the quantum algorithm. We prove this theorem in two steps. First, we follow the proof technique of~\citep{Alon2006} to show that our relaxed notion of approximately feasible is still good enough to ensure a good rounding in expectation. This shows that our notion of feasibility is strong enough for the problem at hand. The stability of the rounding w.r.t. to truncation of the Taylor series then follows by showing appropriate anti-concentration inequalities for the random vector.

Note that in~\cite{Alon2006} the authors prove that the constant $\tfrac{2}{\pi}$ in Proposition~\ref{prop:rounding_main} is optimal.

\subsection{Comparison to existing work} \label{sub:related}

The \textsc{MaxQP SDP} has already received a lot of attention in the literature. Table~\ref{tab:comparison} contains a runtime comparison between the contributions of this work and the best existing classical results \cite{Arora2005,Arora2007}. It highlights regimes, where we obtain both classical and quantum speedups. 
In a nutshell, Hamiltonian Updates outperforms state of the art algorithms whenever the target matrix $A$ has both positive and negative off-diagonal entries and the optimal value of the SDP scales as $n\|A\|$. 
It is worthwhile to explore the following examples.

\subsubsection{Quadratic quantum speedups and classical speedups for generic instances and spin glasses}

Recall that Hamiltonian Updates can only offer speedups for \textsc{MaxQP SDP} instances where the optimal value scales like $n\|A\|$, as opposed to the $\|A\|_{\ell_1}$ scaling required for MMW. Intuitively speaking, such a scaling should arise whenever $A$ has both positive and negative entries, causing cancellations. In order to formalize this intuition, we show that Hamiltonian Updates offers speedups for generic matrices that have both positive and negative entries, see Appendix~\ref{app:randommatrices} for details. Our main result is as follows. Suppose that $A$ is a random Hermitian matrix with entries
\begin{align}\label{equ:SKmodel}
    A_{ij}=\tau_{ij}(g_{ij}+\lambda),
\end{align}
where $g_{ij}$ are independent centered random variables with a bounded fourth moment, $\tau_{ij}$ is a Bernoulli random variable with parameter $p$ and $\lambda>0$ is some fixed parameter. 
This random generative model covers many relevant \textsc{MaxQP} instances.
Note that if we set $p=\frac{s}{n}$, the matrix $A$ is $\cO(s)$ sparse in expectation. Let us first discuss the (centered) $\lambda=0$-case in more detail.
There, $\mathbb{E} \| A \|_{\ell_1} = \Theta (n s)$, $\|A\|_{\infty\to 1}=\Theta(n\sqrt{s})$, $\mathbb{E}\|A\|=\Theta(\sqrt{s})$ and concrete realizations of $A$ concentrate sharply around these expected values.
These concentration arguments are derived in Appendix~\ref{app:randommatrices} and imply that, indeed, $n\|A\|$, and not $\| A \|_{\ell_1}$, provides the right scaling for such generic instances. The scaling for MMW \cite{Arora2005} is
$\tcO(\min\{(n/\epsilon)^{2.5}s, n^{3}\alpha^{-1}\|A\|_{\ell_1}\epsilon^{-3.5}\})$ to achieve an error of $\epsilon\|A\|_{\ell_1}$. Thus, to obtain a multiplicative error for such instances using MMW we need to divide $\epsilon$ by $s^{-\frac{1}{2}}$, yielding an expected scaling of $\tcO(\min\{(n/\epsilon)^{2.5}s^{4.5}, n^{3}s^{2.25}\epsilon^{-3.5}\})$.
This implies that the runtime of Hamiltonian Updates improves upon MMW \cite{Arora2005}, both classically and quantumly.

To the best of our knowledge, the quantum implementation of Hamiltonian Updates establishes the first quantum speedup for problems of this type. Corollary~\ref{cor:quantum} establishes a nearly quadratic speedup for generic \textsc{MaxQP SDP} instances compared to the current state of the art SDP solvers.

It is worth noting that the random matrix defined in Eq.~\eqref{equ:SKmodel} corresponds to a widely studied model in spin glasses: the \emph{(diluted) Sherrington-Kirkpatrick (SK) model}~\cite{Panchenko2013,talagrand_diluted_2011}. 
This problem has received considerable attention in the statistical physics literature. In particular, recent work~\cite{Montanari2018} shows that, under some conjectures, it is possible to approximately solve the quadratic optimization in~\eqref{eq:maxqp} with high probability in time $\tcO(n^2)$ for the standard, undiluted SK model ($\tau_{ij}=1$). This is the same time complexity as our quantum solver, as the target matrix of these instances is dense ($s=\Omega(n)$). To the best of our knowledge, a variation of~\cite{Montanari2018} for the diluted model ($p< 1$) has not yet been discussed.

Furthermore,  there is an integrality gap for the \textsc{SDP} relaxation of this problem in the Gaussian setting \cite{Kunisky2019,montanari_semidefinite_2016} whenever $\lambda=0$. As we discuss in more detail in Appendix~\ref{app:randominstances}, this implies that the value of the problem in the case $\tau_{ij}=1$ converges to the largest eigenvalue of $A$ in the limit $n \to \infty$. On top of that,~\cite{montanari_semidefinite_2016} gives a construction of an approximately optimal feasible point that can be computed in $\cO(n^\omega)$ time,
where $\omega>2$ is the exponent of matrix multiplication. Correspondingly, we do not obtain a classical speedup for such instances. Once again we refer to Appendix~\ref{app:randominstances} for more details and we are not aware of similar results for the diluted model.

Let us now discuss the undiluted case with $\lambda>1$, as the behaviour of the model is not qualitatively different from the case $\lambda=0$ for $\lambda<1$~\cite{montanari_semidefinite_2016}. To the best of our knowledge, the exact value of the \textsc{MaxQP SDP} is not known for this setting. But, numerical evidence suggests that there is no integrality gap~\cite{montanari_semidefinite_2016,Javanmard_2016}, and no constructions of approximately optimal points are known. Thus, we expect that it is this regime, where we obtain both quantum and classical speedups.

\subsubsection{Speedups for \textsc{MaxCut} and the hidden partition problem:} 
Additional structure can substantially reduce the runtime of existing MMW solvers \cite{Arora2007}. For weighted \textsc{MaxCut}, in particular, $A$ is related to the adjacency matrix of a graph and has exclusively non-negative entries. 
This additional structure facilitates the use of powerful dimensionality reduction and sparsification techniques that outperform our algorithm for general graphs. Recently, it was shown that quantum algorithms can speed up spectral graph sparsification techniques~\cite{1911.07306}. As the sparsification step dominates the complexity of these algorithms, this leads to faster solvers for \textsc{MaxCut}, albeit solving the sparsified SDP on a classical computer. 
However, these sparsification techniques do not readily apply to general problem instances, where the entries of $A$ can have both positive and negative signs (sign problem). 
We refer to Appendix~\ref{app:scalingerror} for a more detailed discussion.

More direct speedups do, however, apply for 
approximating \textsc{MaxCut} in Erd\"os-R\'enyi graphs. 
An Erd\"os-R\'enyi graph $G(n,p)$ with $n$ vertices is a random graph in which each edge is present independently at random with probability $p$. In~\cite{dembo_extremal_2017} the authors show that whenever $pn\to\infty$, the \textsc{MaxCut} of such graphs satisfies~:
\begin{align}\label{equ:rightscalingmaxcut2}
\frac{n^2p}{2}+\left(\frac{n^3p(1-p)}{2}\right)^{\frac{1}{2}}P_{*}+o(n^{\frac{3}{2}}),
\end{align}
where $P_*$ is the so-called Parisi constant. One can show that for such graphs, a random balanced partition of the vertices achieves an expected cut of value $\frac{n^2p}{2}$.
Thus, obtaining a cut of $\frac{n^2p}{2}$ up to an approximation error of order $\epsilon \frac{n^2p}{2}$ for $\epsilon$ of constant order is trivial for random graphs: just sample a random one. Thus, obtaining approximations to the \textsc{MaxCut} of such graphs is only interesting whenever we can  achieve an error scaling as $\cO(n^{\frac{3}{2}}\sqrt{p})$ and the usual Goemans-Wiliamson relaxation is not suitable. In order to address this issue, Montanari et al.~\cite{montanari_semidefinite_2016} showed that is advisable to instead solve \textsc{MaxQP} SDP relaxation for the matrix
\begin{align}\label{equ:randomadjacency}
    B=A-p\mathbf{1}^T\mathbf{1},
\end{align}
where $\mathbf{1}$ is the all ones vector and $A$ the (random) adjacency matrix of the graph.
This then has the value $2n^{\frac{3}{2}}\sqrt{p}+o(\sqrt{p/n})$ with high probability. See~[Theorem 1]\cite{montanari_semidefinite_2016} for more details. Note that the matrix in Eq.~\eqref{equ:randomadjacency} has both negative and positive entries with expected value $0$ and bounded variance.
We conclude that we are in the same setting as in the spin glasses for this dense instance and, thus, we also obtain speedups compared to MMW. However, once again the recent work~\cite{Montanari2018} shows that, under some conjectures, it is possible to approximately solve the underlying \textsc{MaxQP} for $B$ directly in time $\cO(n^2)$ with high probability.

Another relevant random graph model is that of the planted partition, whose distribution we will denote by $G(n,a/n,b/n)$ for parameters $a,b>0$. This distribution over graphs with $n$ vertices is defined as follows. First, we partition the $n$ vertices into two subsets $S_1,S_2$ with $|S_1|=n/2$ uniformly at random. Conditional on this partition we pick the edges independently at random with probabilities:
\begin{align*}
\mathbb{P}((i,j)\in E|S_1,S_2)=
\begin{cases}
\tfrac{a}{n},\quad \textrm{if } \{i,j\}\subset S_1 \textrm{ or }  \{i,j\}\subset S_2,\\
\tfrac{b}{n},\quad \textrm{ else}.
\end{cases}
\end{align*}
Solving the \textsc{MaxQP} SDP for the target matrix described in Eq.~\eqref{equ:randomadjacency} with $p=\frac{a+b}{2}$ is relevant to solving the planted partition problem~\cite{montanari_semidefinite_2016} and closely related to the model in Eq.~\eqref{equ:SKmodel} with $\lambda=\frac{a-b}{\sqrt{2(a-b)}}$. We refer to~\cite{montanari_semidefinite_2016} for details on this, but roughly speaking the problem is to decide if a graph was sampled from Erd\"os-R\'enyi distribution with parameter $p$ or from the planted partition with $p=\frac{a+b}{2}$. Note that also for the planted partition the adjacency matrix in Eq.~\eqref{equ:randomadjacency} satisfies the conditions under which we obtain speedups.

In~\cite[Theorem 3]{montanari_semidefinite_2016} the authors show that for certain parameter ranges of $a,b$, solving the \textsc{MaxQP} SDP in Eq.~\eqref{equ:randomadjacency} and using its value to decide which distribution we sampled from gives rise to a good test for this problem. As for both the planted partition and the Erd\"os-R\'enyi model the \textsc{MaxQP} SDP in Eq.~\eqref{equ:randomadjacency} can be solved faster with our methods, we obtain a speedup for this problem.

\subsubsection{Previous quantum SDP solvers:} Previous quantum SDP solvers based on primal-dual methods~\cite{Brandao2017a,Apeldoorn2017,Apeldoorn2018,Brandao2017b} with inverse polynomial dependence  on the error do not provide speedups for solving the \textsc{MaxQP SDP} in the worst case, as their complexity depends on a problem specific parameter, the width of the SDP. We refer to the aforementioned references for a precise definition of this parameter and for the complexity of the solvers under different input models. 
As shown in~\cite[Theorem 24]{Apeldoorn2017}, the width parameter scales at least linearly in the dimension $n$ for what they call \emph{combinable SDPs}~[Definition 23]\cite{Apeldoorn2017}. In a nutshell, these are SDP classes for which direct sum combinations of two instances and constraints yield another valid SDP in the class and we refer to~\cite{Apeldoorn2017} for a precise definition. For our purposes, it suffices to note that the~\textsc{MaxQP SDP} class of SDPs is combinable, as shown in~\cite{Apeldoorn2017}. Although the authors only observe that their conditions apply to \textsc{MaxCut}, it is easy to see that their results do not require any assumptions on the sign of entries of $A$. Thus, their results show that the~\textsc{MaxQP SDP} also admits instances with linearly growing  width. 
To the best of our knowledge, the solvers mentioned above have a dependence that is at least quadratic in the width and at least a $n^{\frac{1}{2}}$ dependence on the dimension. Thus, the combination of the term stemming from the width and the dimension already gives a higher complexity than our solver. One reason why we bypass these restrictions is that we do not use the primal-dual approach to solve the SDP from the aforementioned references.

However, these are worst-case bounds for MAX QP, and do not necessarily imply that our algorithm outperforms the aforementioned ones on the random instances discussed before on average if we pick them uniformly at random. In Prop.~\ref{prop:nosparsedualsolution} of Appendix~\ref{app:randominstances} we show that for the random model in~\eqref{equ:SKmodel} with $\lambda>1$ it is indeed the case that the width scales linearly with the dimension with high probability, albeit under the assumption that the problem does not have an intregrality gap. The absence of an integrality gap is supported by the numerics of~\cite{montanari_semidefinite_2016,Javanmard_2016}. These results show that previous quantum SDP solvers are likely not to provide a speedup on average for such instances with $\lambda>1$. On the other hand, we also show that for $\lambda<1$, the width does not necessarily scale with system size. These results certainly motivate further studies on the width of such randomized instances, also for the random graph models.

Another, and arguably conceptually more interesting, reason why our algorithm outperforms other solvers is how we enforce the diagonal constraint. This is what also sets us apart from other approaches based only on the primal of the SDP, like that of~\cite{Apeldoorn2018}.

Enforcing that each diagonal constraint of the renormalized~\textsc{MAXQP SDP} in Eq.~\eqref{equ:maxqpsdp} is satisfied up to an additive error, i.e.
\begin{align*}
\left|\langle i|\rho|i\rangle-1/n\right|\leq \epsilon
\end{align*}
would require an error $\epsilon$ of order $n^{-1}$ to ensure a solution with a quality comparable to ours. This would translate to an additional quadratic factor in $n$ in the runtime of other approaches, as they have a dependency on the error that is at least inverse quadratic and negate a speedup. This issue with efficiently implementing the diagonal constraints also arises for other approaches to obtain quantum speedups for SDPs beyond those based on quantum Gibbs states. Let us illustrate this by example. In~\cite{Kerenidis2018}, the authors give a quantum SDP solver whose complexity is $\widetilde{O}\left(\frac{n^{2.5}}{\xi^{2}} \mu \kappa^{3} \log \left(\epsilon^{-1}\right)\right)$. Here $\kappa$ and $\mu$ are again problem-specific parameters. On the other hand, $\xi$ is the precision to which each diagonal constraint is satisfied. As noted before, a straightforward implementation of the~\textsc{MAXQP SDP} requires $\xi$ to be at most of order $n^{-1}$, which establishes a runtime of order at least $n^{4.5}$ using those methods. 
 
Alternatively, one could also enforce the diagonal constraint globally through the variational formulation of the trace distance. That would translate to requiring that
\begin{align*}
    \tr{P(\rho-\frac{I}{n})}\leq \epsilon
\end{align*}
holds for all diagonal projectors $P$. As there are $2^n$ such projectors and previous primal-only quantum solvers have a $\sqrt{m}$ dependency on the number of constraints $m=2^n$, this approach also fails to provide a speedup.
Thus, enforcing the diagonal constraints severely limits the scope of existing quantum SDP solvers -- they do not 
readily apply and have worse runtimes than  available classical algorithms. 
Thus, we conclude that all current quantum SDP solvers do not offer speedups over state of the art classical algorithms, see Table~\ref{tab:comparison} for more details.

This discussion showcases that our technique to relax the diagonal constraints gives rise to a novel way of enforcing constraints that allows for better control of errors in quantum SDP solvers and could be used for other relevant SDPs. Moreover, the fact that the approximate solution can still be used to obtain good roundings highlights the fact that our notion of approximate feasibility does not render the problem artificially easy.

\bgroup
\def\arraystretch{1.5}

\begin{table}[]
    \centering
    \small
    \begin{tabular}{|c|c|c|c|c|}
    \hline
         \textbf{Algorithm} & \textbf{Runtime}  & \textbf{Error} & \textbf{Speedup} \\[2pt]
         \hline
         This work (Classical) & $\tcO( \min \{n^2s, n^\omega \}\epsilon^{-12})$ & $\epsilon n\|A\|$ &  -\\[2pt]
         \hline
          This work (Quantum) & $\tcO(n^{1.5}s^{0.5+o(1)}\epsilon^{-28})$ & $\epsilon n\|A\|$ &  -\\[2pt]
         \hline
         MMW~\cite{Arora2005} & $\tcO(\min\{(n/\epsilon)^{2.5}s, n^{3}\alpha^{-1}\|A\|_{\ell_1}\epsilon^{-3.5}\})$ & $\epsilon\|A\|_{\ell_1}$ &$\|A\|_{\ell_1}\geq n\|A\|$, $\epsilon=\Theta(1)$\\[2pt]
         \hline
         Interior-point~\cite{Lee2015} & $\cO(n^{\omega+1}\log(\epsilon^{-1}))$& $\epsilon$  & $\epsilon=\Theta(1)$\\
         
         \hline
         MMW for \textsc{MaxCut}~\cite{Arora2007} 
 & $\tcO(ns)$ & $\epsilon\|A\|_{\ell_1}$  & Erd\"os-R\'enyi random graphs\\
(non-negative entries only) & & & \\
         \hline

    \end{tabular}
    \caption{comparison of different classical algorithms to solve the original \textsc{MaxQP SDP} \eqref{eq:maxqpsdpnormal}. 
The speedup column clarifies in which regimes we obtain speedups and $\omega$ denotes the exponent of matrix multiplication. Here $\alpha$ corresponds to the value of MAXQP SDP.}
    \label{tab:comparison}
\end{table}
\egroup

Finally, we want to point out that subtleties regarding error scaling do not arise for \textsc{MaxCut}. If $A$ is the adjacency matrix of a $d$-regular graph on $n$ vertices, then $n \| A \| = nd = \| A \|_{\ell_1}$ and the different errors in Table~\ref{tab:comparison} all coincide.

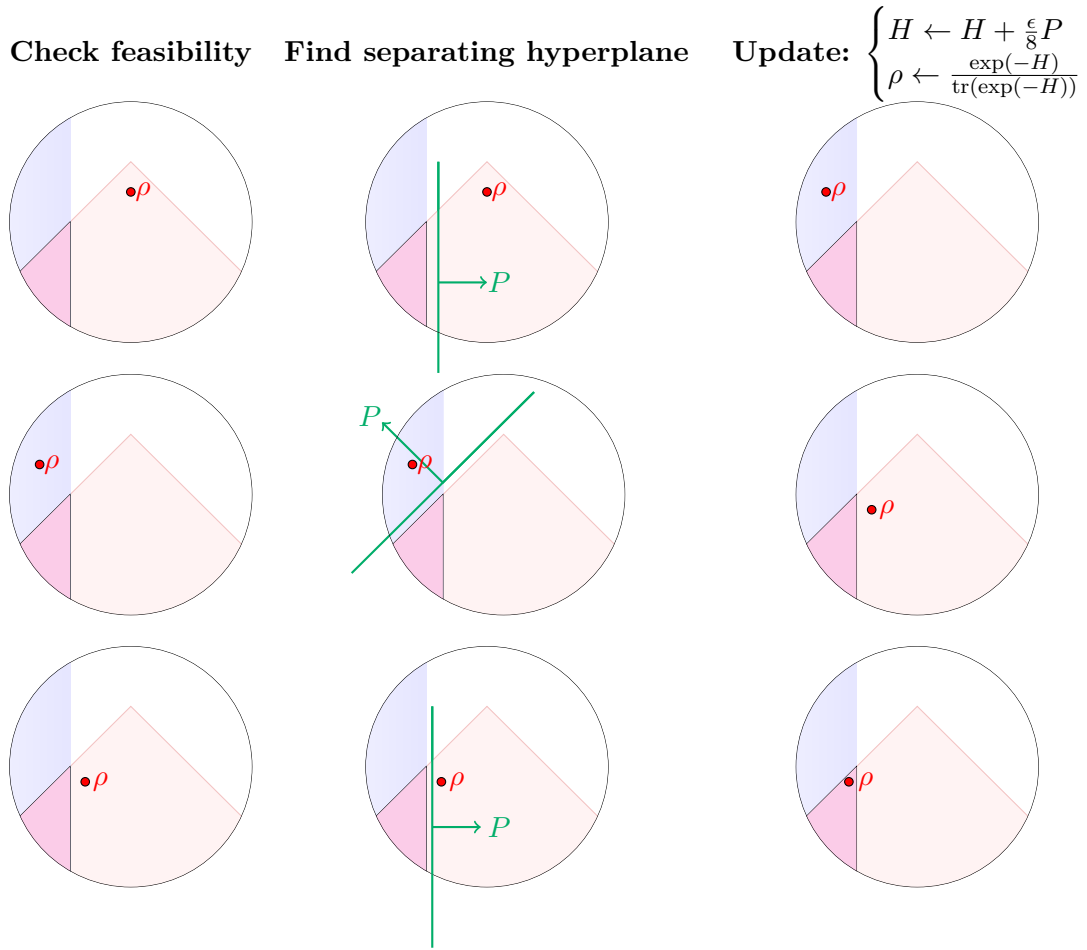
\begin{figure}

\begin{center}
\begin{tabular}{ccc}
\textbf{Check feasibility} & \textbf{Find separating hyperplane} &\textbf{ Update:} 
$\begin{cases}
H \leftarrow  H + \tfrac{\epsilon}{16} P &\\
\rho \leftarrow \frac{\exp (-H)}{\mathrm{tr}(\exp (-H))}&
\end{cases}$ \\
\begin{tikzpicture}[baseline,scale=0.8]
\clip (0,1) circle (2cm);
\fill[white] (2,0) -- (0,2) -- (-2,0) -- (0,-2) -- (2,0);
\draw[red,opacity=0.2] (2,0) -- (0,2) -- (-2,0) -- (0,-2) -- (2,0);
\draw[fill=red, opacity =0.05] (2,0) -- (0,2) -- (-2,0) -- (0,-2) -- (2,0);
\draw[blue,opacity=0.2] (-1,-3) -- (-1,3);
\shade[left color=white, right color=blue!10] (-4,-3) rectangle (-1,3);
\draw (-2,0) -- (-1,1) -- (-1,-1) -- (-2,0);
\fill[white] (-2,0) -- (-1,1) -- (-1,-1) -- (-2,0);
\fill[magenta, opacity=0.2] (-2,0) -- (-1,1) -- (-1,-1) -- (-2,0);
\draw[fill=red] (0,1.5) circle (2pt);
\node at (0.2,1.5) {\textcolor{red}{$\rho$}};
\draw (0,1) circle (2cm);
\end{tikzpicture}
&
\begin{tikzpicture}[baseline,scale=0.8]
\draw[thick,Green] (-0.8,-1.5) -- (-0.8,2);
\clip (0,1) circle (2cm);
\fill[white] (2,0) -- (0,2) -- (-2,0) -- (0,-2) -- (2,0);
\draw[red,opacity=0.2] (2,0) -- (0,2) -- (-2,0) -- (0,-2) -- (2,0);
\draw[fill=red, opacity =0.05] (2,0) -- (0,2) -- (-2,0) -- (0,-2) -- (2,0);
\draw[blue,opacity=0.2] (-1,-3) -- (-1,3);
\shade[left color=white, right color=blue!10] (-4,-3) rectangle (-1,3);
\draw (-2,0) -- (-1,1) -- (-1,-1) -- (-2,0);
\fill[white] (-2,0) -- (-1,1) -- (-1,-1) -- (-2,0);
\fill[magenta, opacity=0.2] (-2,0) -- (-1,1) -- (-1,-1) -- (-2,0);
\draw[fill=red] (0,1.5) circle (2pt);
\node at (0.2,1.5) {\textcolor{red}{$\rho$}};
\draw[thick,Green] (-0.8,-2) -- (-0.8,2);
\draw[thick,Green,->] (-0.8,0) -- (0,0);
\node at (0.2,0) {\textcolor{Green}{$P$}};
\draw (0,1) circle (2cm);
\end{tikzpicture}
&
\begin{tikzpicture}[baseline,scale=0.8]
\clip (0,1) circle (2cm);
\fill[white] (2,0) -- (0,2) -- (-2,0) -- (0,-2) -- (2,0);
\draw[red,opacity=0.2] (2,0) -- (0,2) -- (-2,0) -- (0,-2) -- (2,0);
\draw[fill=red, opacity =0.05] (2,0) -- (0,2) -- (-2,0) -- (0,-2) -- (2,0);
\draw[blue,opacity=0.2] (-1,-3) -- (-1,3);
\shade[left color=white, right color=blue!10] (-4,-3) rectangle (-1,3);
\draw (-2,0) -- (-1,1) -- (-1,-1) -- (-2,0);
\fill[white] (-2,0) -- (-1,1) -- (-1,-1) -- (-2,0);
\fill[magenta, opacity=0.2] (-2,0) -- (-1,1) -- (-1,-1) -- (-2,0);
\draw[fill=red] (-1.5,1.5) circle (2pt);
\node at (-1.3,1.5) {\textcolor{red}{$\rho$}};
\draw (0,1) circle (2cm);
\end{tikzpicture}
 \\
\begin{tikzpicture}[baseline,scale=0.8]
\draw[white] (0,-1.5) -- (1,-1.5); %
\clip (0,1) circle (2cm);
\fill[white] (2,0) -- (0,2) -- (-2,0) -- (0,-2) -- (2,0);
\draw[red,opacity=0.2] (2,0) -- (0,2) -- (-2,0) -- (0,-2) -- (2,0);
\draw[fill=red, opacity =0.05] (2,0) -- (0,2) -- (-2,0) -- (0,-2) -- (2,0);
\draw[blue,opacity=0.2] (-1,-3) -- (-1,3);
\shade[left color=white, right color=blue!10] (-4,-3) rectangle (-1,3);
\draw (-2,0) -- (-1,1) -- (-1,-1) -- (-2,0);
\fill[white] (-2,0) -- (-1,1) -- (-1,-1) -- (-2,0);
\fill[magenta, opacity=0.2] (-2,0) -- (-1,1) -- (-1,-1) -- (-2,0);
\draw[fill=red] (-1.5,1.5) circle (2pt);
\node at (-1.3,1.5) {\textcolor{red}{$\rho$}};
\draw (0,1) circle (2cm);
\end{tikzpicture}
&
\begin{tikzpicture}[baseline,scale=0.8]
\node at (-2.2,2.3) {$\textcolor{Green}{P}$};
\draw[thick,Green] (-2.5,-0.5+0.2) -- (0.5,2.5+0.2);
\draw[thick,Green,->] (-1,1+0.2) -- (-2,2+0.2);
\clip (0,1) circle (2cm);
\fill[white] (2,0) -- (0,2) -- (-2,0) -- (0,-2) -- (2,0);
\draw[red,opacity=0.2] (2,0) -- (0,2) -- (-2,0) -- (0,-2) -- (2,0);
\draw[fill=red, opacity =0.05] (2,0) -- (0,2) -- (-2,0) -- (0,-2) -- (2,0);
\draw[blue,opacity=0.2] (-1,-3) -- (-1,3);
\shade[left color=white, right color=blue!10] (-4,-3) rectangle (-1,3);
\draw (-2,0) -- (-1,1) -- (-1,-1) -- (-2,0);
\fill[white] (-2,0) -- (-1,1) -- (-1,-1) -- (-2,0);
\fill[magenta, opacity=0.2] (-2,0) -- (-1,1) -- (-1,-1) -- (-2,0);
\draw[thick,Green] (-2.5,-0.5+0.2) -- (0.5,2.5+0.2);
\draw[thick,Green,->] (-1,1+0.2) -- (-2,2+0.2);
\node at (-2.1,2.1+0.2) {\textcolor{Green}{$S$}};
\draw[fill=red] (-1.5,1.5) circle (2pt);
\node at (-1.3,1.5) {\textcolor{red}{$\rho$}};
\draw (0,1) circle (2cm);
\end{tikzpicture}
&
\begin{tikzpicture}[baseline,scale=0.8]
\clip (0,1) circle (2cm);
\fill[white] (2,0) -- (0,2) -- (-2,0) -- (0,-2) -- (2,0);
\draw[red,opacity=0.2] (2,0) -- (0,2) -- (-2,0) -- (0,-2) -- (2,0);
\draw[fill=red, opacity =0.05] (2,0) -- (0,2) -- (-2,0) -- (0,-2) -- (2,0);
\shade[left color=white, right color=blue!10] (-4,-3) rectangle (-1,3);
\draw (-2,0) -- (-1,1) -- (-1,-1) -- (-2,0);
\fill[white] (-2,0) -- (-1,1) -- (-1,-1) -- (-2,0);
\fill[magenta, opacity=0.2] (-2,0) -- (-1,1) -- (-1,-1) -- (-2,0);
\draw[fill=red] (-0.75,0.75) circle (2pt);
\node at (-0.5,0.75) {\textcolor{red}{$\rho$}};
\draw (0,1) circle (2cm);
\end{tikzpicture}
\\
\begin{tikzpicture}[baseline,scale=0.8]
\draw[white] (0,-1.5) -- (1,-1.5); %
\clip (0,1) circle (2cm);
\fill[white] (2,0) -- (0,2) -- (-2,0) -- (0,-2) -- (2,0);
\draw[red,opacity=0.2] (2,0) -- (0,2) -- (-2,0) -- (0,-2) -- (2,0);
\draw[fill=red, opacity =0.05] (2,0) -- (0,2) -- (-2,0) -- (0,-2) -- (2,0);
\shade[left color=white, right color=blue!10] (-4,-3) rectangle (-1,3);
\draw (-2,0) -- (-1,1) -- (-1,-1) -- (-2,0);
\fill[white] (-2,0) -- (-1,1) -- (-1,-1) -- (-2,0);
\fill[magenta, opacity=0.2] (-2,0) -- (-1,1) -- (-1,-1) -- (-2,0);
\draw[fill=red] (-0.75,0.75) circle (2pt);
\node at (-0.5,0.75) {\textcolor{red}{$\rho$}};
\draw (0,1) circle (2cm);
\end{tikzpicture}
&
\begin{tikzpicture}[baseline,scale=0.8]
\draw[white] (0,-1.5) -- (1,-1.5); %
\draw[thick,Green] (-0.9,-2) -- (-0.9,2);
\clip (0,1) circle (2cm);
\fill[white] (2,0) -- (0,2) -- (-2,0) -- (0,-2) -- (2,0);
\draw[red,opacity=0.2] (2,0) -- (0,2) -- (-2,0) -- (0,-2) -- (2,0);
\draw[fill=red, opacity =0.05] (2,0) -- (0,2) -- (-2,0) -- (0,-2) -- (2,0);
\shade[left color=white, right color=blue!10] (-4,-3) rectangle (-1,3);
\draw (-2,0) -- (-1,1) -- (-1,-1) -- (-2,0);
\fill[white] (-2,0) -- (-1,1) -- (-1,-1) -- (-2,0);
\fill[magenta, opacity=0.2] (-2,0) -- (-1,1) -- (-1,-1) -- (-2,0);
\draw[fill=red] (-0.75,0.75) circle (2pt);
\node at (-0.5,0.75) {\textcolor{red}{$\rho$}};
\draw (0,1) circle (2cm);
\draw[thick,Green] (-0.9,-2) -- (-0.9,2);
\draw[thick,Green,->] (-0.9,0) -- (-0.1,0);
\node at (0.2,0) {\textcolor{Green}{$P$}};
\end{tikzpicture}
&
\begin{tikzpicture}[baseline,scale=0.8]
\draw[white] (0,-1.5) -- (1,-1.5); %
\clip (0,1) circle (2cm);
\fill[white] (2,0) -- (0,2) -- (-2,0) -- (0,-2) -- (2,0);
\draw[red,opacity=0.2] (2,0) -- (0,2) -- (-2,0) -- (0,-2) -- (2,0);
\draw[fill=red, opacity =0.05] (2,0) -- (0,2) -- (-2,0) -- (0,-2) -- (2,0);
\shade[left color=white, right color=blue!10] (-4,-3) rectangle (-1,3);
\draw (-2,0) -- (-1,1) -- (-1,-1) -- (-2,0);
\fill[white] (-2,0) -- (-1,1) -- (-1,-1) -- (-2,0);
\fill[magenta, opacity=0.2] (-2,0) -- (-1,1) -- (-1,-1) -- (-2,0);
\draw[fill=red] (-1.125,0.75) circle (2pt);
\node at (-0.85,0.75) {\textcolor{red}{$\rho$}};
\draw (0,1) circle (2cm);
\end{tikzpicture}
\end{tabular}
\end{center}
\caption{\textit{Caricature of Hamiltonian Update iterations in Algorithm~\ref{alg:HU}:} Schematic illustration of the intersection of three convex sets (i) a halfspace (blue), (ii) a diamond-shaped convex set (red) and (iii) the set of all quantum states (clipped circle). Algorithm~\ref{alg:HU} (Hamiltonian Updates) approaches a point in the convex intersection (magenta) of all three sets by iteratively checking feasibility (left column), identifying a separating hyperplane (central column) and updating the matrix exponent to penalize infeasible directions (right column).}
\label{fig:HU}
\end{figure}

\section{Technical details and proofs}

\subsection{Proof of Theorem~\ref{thm:main}} \label{sec:proof_main}

By construction, Algorithm~\ref{alg:HU} (Hamiltonian Updates) terminates as soon as it has found a quantum state $\rho$ that is $\epsilon$-close to being feasible. 
Correctly flagging infeasibility is the more interesting aspect of Theorem~\ref{thm:main} (convergence to feasible point). Several variations of the statement and proof below can be found in the literature~\cite{Tsuda2005,Hazan2006,Arora2007,Steurer2015,Bubeck2015,Brandao2017b,Aaronson_2019}, but we present it for completeness.

\begin{lem}\label{thm:algocorrect}
Suppose Algorithm~\ref{alg:HU} does not terminate after $T=\lceil 64 \log (n)/\epsilon^2 \rceil+1$ steps. Then, the feasibility problem \eqref{eq:CPfeas} is infeasible.
\end{lem}

\begin{proof}
By contradiction.
Suppose there exists a feasible point $\rho^*$ in the intersection of all $m+1$ sets and we ran the algorithm for $T$ steps. 
Instantiate the short-hand notation $\rho_t = \rho_{H_t} = \exp (-H_t)/\mathrm{tr}(\exp (-H_t))$ for the $t$-th state and Hamiltonian in Algorithm~\ref{alg:HU}. 
Initialization with $H_0 = 0$ and $\rho_0 = I/n$ is crucial, as it implies that the quantum relative entropy between $\rho^*$ and $\rho_0$ is bounded:
\begin{equation*}
\rl{\rho^*}{\rho_0} = \tr{\rho^*\lb \log \rho^* - \log \rho_0 \rb}
\leq \log (n).
\end{equation*}
We will now show that the relative entropy between successive (infeasible) iterates $\rho_{t+1},\rho_{t}$ and the feasible state $\rho^*$ necessarily decreases by a finite amount. 
Let $P_t$ be the hyperplane that separates $\rho_t$ from the feasible set provided by the oracle. The update rule $H_{t+1}=H_{t} + \tfrac{\epsilon}{16}P_t$ then asserts
\begin{align}\label{equ:rel_entropy_diff}
\rl{\rho^*}{\rho_{t+1}} - \rl{\rho^*}{\rho_t}
=& \tr{ \rho^* (H_{t+1}-H_{t})} + \log \left( \frac{\tr{ \exp (-H_{t+1})}}{ \tr{ \exp (-H_t)}} \right)\nonumber \\
=& \tfrac{\epsilon}{16} \tr{ P_t \rho^*} - \log \left( \frac{\tr{ \exp \left(-H_{t+1}+\tfrac{\epsilon}{16}P_t \right)}}{\tr{\exp (-H_{t+1})}}\right).
\end{align}
The logarithmic ratio can be bounded using the Peierls-Bogoliubov inequality \cite[Lemma 1]{Araki1970}: $\log \left( \tr{ \exp (F+G)}\right) \geq \tr{F \exp (G)}$ provided that $\tr{\exp (G)}=1$. This implies
\begin{align}
 \log \left( \frac{\tr{ \exp \left(-H_{t+1}+\tfrac{\epsilon}{16}P_t \right)}}{\tr{\exp (-H_{t+1})}}\right) =&  \log \left( \tr{\exp \left( -H_{t+1} - \log \left( \tr{ \exp (-H_{t+1})} I+\tfrac{\epsilon}{16} P_t \right) \right)}   \right) \nonumber\\
\geq &\tr{ \tfrac{\epsilon}{16} P_t \exp \left( -H_{t+1} - \log ( \tr{\exp (-H_{t+1})})I\right)} \nonumber\\
    =& \tfrac{\epsilon}{16} \tr{ P_t \exp (-H_{t+1})/\tr{\exp (-H_{t+1})}} = \tfrac{\epsilon}{16} \tr{ P_t \rho_{t+1}}. \label{equ:peierls_bogoliubov}
\end{align}
Combining Eq.~\eqref{equ:rel_entropy_diff} with Eq.~\eqref{equ:peierls_bogoliubov} we arrive at
\begin{align*}
\rl{\rho^*}{\rho_{t+1}} - \rl{\rho^*}{\rho_t}\leq \tfrac{\epsilon}{16} \tr{ P_t (\rho^*-\rho_{t+1})}.
\end{align*}
Next, note that the updates are mild in the sense that $\rho_{t+1}$ and $\rho_t$ are close in trace distance. \cite[Lem.~16]{Brandao2017a} implies 
$\| \rho_{t_1} - \rho_t \|_{tr} \leq 2 \left( \exp (\tfrac{\epsilon}{16} \|P_t \|)-1\right) \leq \tfrac{\epsilon}{4}$, because $\| P_t \| \leq 1$ by construction and we can also assume $\tfrac{\epsilon}{16} \leq \log (2)$. 
Combining these insights with Matrix H\"older \cite[Ex.~IV.2.12]{Bhatia1997} ensures
\begin{align*}
\rl{\rho^*}{\rho_{t+1}} - \rl{\rho^*}{\rho_t} \leq & \tfrac{\epsilon}{16} \tr{P_t \rho^*} - \tfrac{\epsilon}{16} \tr{ P_t \rho_{t+1}} \\
=& \tfrac{\epsilon}{16} \left( \tr{ P_t \left( \rho_{t}-\rho_{t+1} \right)} - \tr{P_t \left( \rho_t - \rho^*\right)} \right) \\
\leq & \tfrac{\epsilon}{16} \left( \| P_t \| \| \rho_{t}-\rho_{t+1} \|_{tr} - \tr{P_t \left( \rho_t-\rho^*\right)} \right).
\end{align*}
The first contribution is bounded by $\tfrac{\epsilon}{4}\| P_t \| \leq \tfrac{\epsilon}{4}$, while Definition~\ref{def:oracle} ensures $\tr{P_t (\rho_t-\rho^*)} \geq \frac{\epsilon}{2}$ ($\rho^*$ is feasible and $P_t$ is an $\frac{\epsilon}{2}$-separation oracle for the infeasible point $\rho_t$). In summary,
\begin{equation*}
\rl{\rho^*}{\rho_{t+1}} - \rl{\rho^*}{\rho_t} \leq \tfrac{\epsilon}{16} \left( \tfrac{\epsilon}{4} - \tfrac{\epsilon}{2} \right) = - \tfrac{\epsilon^2}{64}
\quad \textrm{for all iterations} \quad t=0,\ldots,T
\end{equation*}
and we conclude
\begin{equation*}
\rl{\rho^*}{\rho_T}= \sum_{t=0}^T \left( \rl{\rho^*}{\rho_{t+1}} - \rl{\rho^*}{\rho_t}\right)+ \rl{\rho^*}{\rho_0} \leq - T \tfrac{\epsilon^2}{64} + \log (n).
\end{equation*}
This expression becomes negative as soon as the total number of steps $T$ surpasses $64 \log (n)/\epsilon^2$. A contradiction, because the quantum relative entropy is always non-negative.
\end{proof}
\subsection{Stability of the relaxed \textsc{MaxQP SDP}}
Note that even if Algorithm~\ref{alg:HU} accepts a candidate point, it does not necessarily mean that this point is exactly feasible. Theorem~\ref{thm:main} only asserts that this point is $\epsilon$-close to all sets of interest with respect to a set of observables. For the \textsc{MaxQP SDP} \eqref{equ:maxqpsdp}, this means that the outputs of the algorithm will only satisfy the diagonal constraints approximately
 and, in principle, the value of this further relaxed problem could differ significantly from the original value. In the next Proposition we show that this is not the case:

\begin{prop}\label{prop:stability}
Let $\alpha_{\epsilon^4}=\tr{A\rho}$ be the value 
attained by  -- up to accuracy $\epsilon^4$ -- solution $\rho$ to the relaxed \textsc{MaxQP SDP
}~\eqref{equ:maxqpsdp} with input matrix $A$. Then there is a quantum state $\rho^{\sharp}$ at trace distance $\cO(\epsilon)$ of $\rho$ such that $n\rho^{\sharp}$ is a feasible point of \textsc{MaxQP SDP
}~\eqref{eq:maxqpsdpnormal}. In particular:
\begin{align}\label{equ:promisevalue}
\left|\alpha_{\epsilon^4} n\|A\|-\tr{n\rho^{\sharp}A}\right|=\cO(\epsilon n\|A\|),
\end{align}
Moreover, it is possible to construct $\rho^{\sharp}$ in time $\cO(n^2)$ given the entries of $\rho$.
\end{prop}
\begin{proof}
Let $\rho$ be a solution to the relaxed \textsc{MaxQP SDP}~\eqref{equ:maxqpsdp} with relaxation parameter $\epsilon^4$. We will now construct $\rho^\sharp$ such  that $n\rho^\sharp$ is an exactly feasible point of the \textsc{MaxQP SDP}~\eqref{equ:maxqpsdp}. These modifications are mild enough to ensure that the associated SDP value will only change by $\cO(\epsilon n\|A\|)$. We proceed in two steps: (i) $\rho \mapsto \rho'$: Identify diagonal entries that substantially deviate from $1/n$ in the sense that $| \langle i| \rho| i \rangle - 1/n| >\epsilon^2/n$. Subsequently, replace $\rho_{ii}$ by $1/n$ and set all entries in the $i$-th row and $i$-th column to zero. This ensures that $\rho'$ remains positive semidefinite.
(ii) $\rho' \to R$: Replace all remaining diagonal entries by $1/n$. This may thwart positive semidefiniteness, but the following convex combination restores this feature:
\begin{equation*}
\rho^\sharp = \tfrac{1}{1+\epsilon^2} \left( R + \tfrac{\epsilon^2}{n} I \right).
\end{equation*}
By construction, this matrix is both psd and obeys $\langle i| \rho^\sharp |i \rangle=1/n$ for all $i \in \left[n\right]$. In words: it is a feasible point of the renormalized \textsc{MaxQP SDP} \eqref{equ:maxqpsdp}. 

We now show that these reformulations are mild. To this end, let $B=\left\{i:\; \left|n\langle i| \rho |i \rangle - 1\right| > \epsilon^2\right\} \subset \left[n\right]$ be the indices associated with large deviations. Without loss of generality, we can assume that these are the first $|B|$ indices. 
Then,
\begin{align}
\| \rho' - \rho \|_{tr}
=& \left\|
\left( \begin{array}{cc}
n^{-1} I_B & 0 \\
0 & \rho_{22}
\end{array}
\right) - \left(
\begin{array}{cc}
\rho_{11} & \rho_{12} \\
\rho_{21} & \rho_{22}
\end{array}
\right)
\right\|_{tr}
= \left\|
\left(
\begin{array}{cc}
n^{-1} I_B -\rho_{11}& - \rho_{12} \\
-\rho_{21} & 0
\end{array}
\right)
\right\|_{tr} \nonumber \\
\leq &\| \rho_{11}\|_{tr} + 2 \| \rho_{12} \|_{tr} + \| n^{-1} I_B \|_{tr}. \label{eq:stability_aux1}
\end{align}
Next, note that $\epsilon^4$-approximate feasibility implies $\sum_{i=1}^n \left| \langle i| \rho |i \rangle - 1/n \right| \leq \epsilon^4$. This, in turn, demands $|B| \tfrac{\epsilon^2}{n} \leq \epsilon^4$ or, equivalently $|B| \leq n \epsilon^2$. The definition of $B$ moreover asserts
\begin{equation*}
\| \rho_{22} \|_{tr} \geq (n-|B|) \tfrac{1-\epsilon^2}{n} \geq (1-\epsilon^2)^2.
\end{equation*}
Moreover, as shown in~\cite{King2003}, we have
\begin{equation*}
\left\|\left[
\begin{array}{c c}
\| \rho_{11}\|_{tr} &  \|\rho_{12}\|_{tr}\\ 
\|\rho_{12}^T\|_{tr} & \|\rho_{22}\|_{tr}
\end{array}\right]\right\|_{tr} \leq
\left\|\left[
\begin{array}{c c}
 \rho_{11} &  \rho_{12}\\ 
\rho_{12}^T & \rho_{2}
\end{array}\right]\right\|_{tr} = \| \rho \|_{tr} = \mathrm{tr}(\rho)=1.
\end{equation*}
As $\|\cdot\|_{tr}\geq\|\cdot\|_2$ (the Frobenius, or Schatten-2 norm), it follows from the last equation that
\begin{align*}
\| \rho_{11}\|_{tr}^2+2\|\rho_{12}\|_{tr}^2+\|\rho_{22}\|_{tr}^2\leq1.
\end{align*}
And, as
 $\|\rho_{22}\|_{tr}^2\geq (1-\epsilon^2)^4$, we conclude
$
\| \rho_{11}\|_{tr}^2+2\|\rho_{12}\|_{tr}^2=\cO(\epsilon^2).
$ which in turn implies $\| \rho_{11} \|_{tr} + 2 \| \rho_{12} \|_{tr} = \cO (\epsilon)$.
Inserting this relation into Eq.~\eqref{eq:stability_aux1} yields
\begin{equation}\label{equ:rhoandprimeclose}
\| \rho' - \rho \|_{tr} = \cO (\epsilon).
\end{equation}
Next, note that we obtain $R$ from $\rho'$ by just replacing all diagonal entries of $\rho'$ by $1/n$. As by construction all the diagonal elements of $\rho'$ are in the range $1/n\pm\epsilon^2/n$, we can write
\begin{align*}
    R=\rho'+D,
\end{align*}
where $D$ is a diagonal matrix whose entries are in the range $[-\epsilon^2/n,\epsilon^2/n]$. Thus, $D+\frac{\epsilon^2}{n}I$ is psd. Normalizing the trace we see that
\begin{align*}
   \rho^{\sharp}=\frac{1}{1+\epsilon^2}\left(\rho'+D+\frac{\epsilon^2}{n}I\right)
\end{align*}
is psd with diagonal entries $1/n$ and, thus, $n\rho^{\#}$ is a feasible point of \textsc{MaxQP SDP}~\eqref{eq:maxqpsdpnormal}.
We also have that:
\begin{align}\label{equ:sharpclosetoprime}
    \|\rho'-\rho^{\sharp}\|_{tr}=\frac{1}{1+\epsilon^2}\|\epsilon^{2}\rho'+D+\epsilon^2\frac{I}{n}\|_{tr}=\cO(\epsilon^2).
\end{align}
by a triangle inequality.
Thus, combining Eq.~\eqref{equ:sharpclosetoprime} and Eq.~\eqref{equ:rhoandprimeclose} we conclude from another triangle inequality that:
\begin{align*}
    \|\rho-\rho^{\sharp}\|_{tr}=\cO(\epsilon).
\end{align*}
The claim then follows from a (matrix) H\"older inequality:
\begin{align*}
 \left|\tr{n A \rho}-\tr{n A \rho^{\sharp}}\right|\leq n\|A\|\|\rho-\rho^{\sharp}\|_{tr}=\cO(n\|A\|\epsilon).
\end{align*}
Note that the proof technique above is constructive and allows us to construct a feasible point from an approximately feasible one in $\cO(n^2)$ time by manipulating the entries.
\end{proof}

\subsection{Approximately solving the \textsc{MaxQP SDP} on a classical computer} \label{sec:classical}

We will now show how to use Hamiltonian Updates (Algorithm~\ref{alg:HU}) to solve the \textsc{MaxQP SDP} \eqref{equ:maxqpsdp} on a classical computer.
It turns out that the main classical bottleneck is the cost of computing matrix exponentials $\rho = \exp (-H) / \tr{\exp (-H)}$. 
The following result, also observed in~\cite{Steurer2015}, asserts that coarse truncations of the matrix exponential already yield accurate approximations.

\begin{lem}\label{lem:errortruncation}
Fix a Hermitian $n \times n$  matrix $H$, an accuracy $\epsilon$ and let $l$ be the smallest even number that obeys $(l+1)(\log (l+1)-1) \geq 2\| H \|+ \log (n)+ \log (1/\epsilon)$. Then, the truncated matrix exponential $T_l = \sum_{k=0}^l \tfrac{1}{k!} (-H)^k$ is guaranteed to obey
\begin{equation*}
\left\| \frac{\exp (-H)}{\tr{\exp (-H)}} - \frac{T_l}{ \tr{T_l}} \right\|_{tr} \leq \epsilon.
\end{equation*}
\end{lem}

\begin{proof}
First note, that truncation at an even integer $l$ ensures that $T_l$ is positive semidefinite. This is an immediate consequence of the fact that even-degree Taylor expansions of the (scalar) exponential are non-negative polynomials. In particular, $\| T_l \|_{tr} = \tr{T_l}$.
Combine this with $\tr{X} \leq \| X \|_{tr} \leq n \| X \|$ for all Hermitian $n \times n$ matrices to conclude
\begin{align*}
\left\| \frac{\exp (-H)}{\tr{\exp (-H)}} - \frac{T_l}{\tr{T_l}} \right\|_{tr}
\leq & \frac{1}{\tr{\exp(-H)}} \left\| \exp (-H) - T_l \right\|_{tr} + \frac{\left| \tr{ \exp (-H)} - \tr{ T_l} \right|}{\tr{ T_l} \tr{ \exp (-H)}} \| T_l \|_{tr} \\ \\
\leq & \frac{2  \| \exp (-H) - T_l \|_{tr}}{\tr{ \exp (-H)}} 
\leq 2n \exp (\|H \|) \| \exp (-H) - T_l \|,
\end{align*}
where we have also used $\tr{\exp (-H)} \geq \| \exp (-H) \| \geq \exp (- \|H \|)$. By construction, both $\exp (-H)$ and $T_l$ commute and are diagonal in the same eigenbasis. Let $\lambda_1,\ldots, \lambda_n$ be the eigenvalues of $H$. Then, Taylor's remainder theorem asserts
\begin{align*}
\| \exp (-H) - T_l \| = \max_{1 \leq i \leq n} \left| \exp (-\lambda_i) - \sum_{k=0}^l \tfrac{1}{k!} (-\lambda)^k \right| 
\leq \frac{\max_{i}\exp \left(- \lambda_i\right)}{(l+1)!}
\leq \frac{ \exp (\|H \|)}{(l+1)!}.
\end{align*}
The value of $l$ is chosen such that
\begin{equation*}
\frac{2n \exp (2\|H \|)}{(l+1)!}
\leq \exp (2 \| H \| + \log (2) + \log (n) - 1 - (l+1)(\log (l+1)-1) )
\leq \epsilon,
\end{equation*}
because $(l+1)! \geq \mathrm{e} \left( (l+1)/\mathrm{e}\right)^{l+1}$.
\end{proof}

\begin{cor}
Given an $s$ sparse, symmetric $n \times n$ matrix $A$ and $\epsilon>0$, we can solve the \textsc{MaxQP SDP}~\eqref{equ:maxqpsdp} up to an additive error $\cO(\epsilon n\|A\|)$ in time $\tcO(\min\{n^2s,n^{\omega}\}\epsilon^{-12})$ on a classical computer.
\end{cor}

Although the dependency in $\epsilon$ for our algorithm is high, we expect that a more refined analysis of the error could improve this significantly.
This is because the approximately feasible to feasible conversion behind Proposition~\ref{prop:stability} requires $\epsilon^4$ accuracy.

\begin{proof}
As each run of Algorithm~\ref{alg:HU} takes at most $\tcO(1)$ iterations, we only need to implement the oracles in time $\tcO(n^2s\epsilon^{-1})$ to establish the advertised runtime for an approximate solution. 
First, note that the operator norm $\| H_t \|$ only grows modestly with the number of iterations $t=0,\ldots,T$. This readily follows from $H_0 = 0$, and $\| H_{t+1}-H_t \| \leq \tfrac{\epsilon}{16} \|P_t \| \leq \tfrac{\epsilon}{16}$. What is more, the maximal number of steps is $T=\lceil 64 \log (n)/\epsilon^2\rceil$, implying $\| H_t \| \leq 4 \log (n) /\epsilon$ for all $t$.

In turn, Lemma~\ref{lem:errortruncation} implies that computing the Taylor series of $\exp (-H_t)$ up to a term of order $\mathcal{O}(\log (n)/\epsilon)$ suffices to compute a matrix $\tilde{\rho}_t$ that is $\tfrac{\epsilon}{4}$-close to the true iterate $ \rho_t=\exp (-H_t)/\tr{\exp (-H_t)}$ in trace distance.  Now note that the complexity of multiplying any matrix with $H_t$ is $\cO(\min\{n^2s,n^{\omega}\})$, as $H_t$ is a linear combination of a diagonal matrix and $A$. Thus, we conclude that computing $\tilde{\rho}_t$ takes time $\cO(n^{2}s\log(n)/\epsilon)$. Checking the diagonal constraints then takes time $\cO(n)$ and computing $\tr{A\|A\|^{-1}\tilde{\rho}_t}$ takes time $\cO(ns)$. This suffices to implement both $\epsilon$-separation oracles and highlights that the runtime is dominated by computing approximations of the matrix exponential.

 Finally, we show in Proposition~\ref{prop:stability} that in order to ensure an additive error of order $\cO(\epsilon n\|A\|)$ for the \textsc{MaxQP SDP}, it suffices to solve the relaxed one up to an error $\epsilon^4$, from which the claim follows and we can then convert the approximately feasible solution to a feasible solution in time $\cO(n^2)$.
\end{proof}

\subsection{Approximately solving the \textsc{MaxQP SDP} on a quantum computer} \label{sec:quantum}
We will now show how to implement $\epsilon$-separation oracles on a quantum computer.  As discussed before, implementing the oracle requires us to evaluate diagonals of the Gibbs states $\rho= \exp (-H)/\tr{\exp (-H)}$ and the value of $\tr{\rho A\|A\|^{-1}}$. These two tasks can be performed easily on a quantum computer given the ability to prepare approximate copies of the quantum state $\rho$.

\begin{lem}\label{thm:quantumoracle}
We can implement $\epsilon$-separation oracles for the \textsc{MaxQP SDP} \eqref{equ:maxqpsdp} on a quantum computer given access to $\tfrac{\epsilon}{8}$ approximate $\cO(n\epsilon^{-2})$ copies in trace distance of the input state $\rho$ and the ability to measure $\tr{A\rho}\|A\|^{-1}$. Moreover, the classical postprocessing time needed to implement the oracle is $\cO(n\epsilon^{-2})$.
\end{lem}

\begin{proof}
Let $\tilde{\rho}$ be the approximation to $\rho$.
We implement the oracle by first measuring $\cO(n\epsilon^{-2})$ approximate copies $\tilde{\rho}$ of the input $\rho$ in the computational basis.  This is enough to ensure that with high probability the resulting empirical distribution of the measurement outcomes, $\hat{p}=\sum_i\hat{p}(i)\ketbra{i}$, satisfies
\begin{align*}
\|\sum_i\langle i|\tilde{\rho}|i\rangle\ketbra{i}-\hat{p}\|_{tr}\leq\frac{\epsilon}{8}.
\end{align*}
If $\|I/n-\hat{p}\|_{tr}\leq \frac{3\epsilon}{4}$, then the oracle for the diagonal constraints accepts the current state. If not, we output
$$
P= \sum_{i=1}^n (\mathbb{I} \left\{ \tilde{p}(i) >1/n \right\}-\left\{ \tilde{p}(i) <1/n \right\}) |i \rangle \! \langle i|.
$$
To see that this indeed satisfies the definition of the oracle, note that the empirical distribution $\hat{p}$ is at most $\frac{\epsilon}{4}$ away in total variation distance to the distribution on the diagonals of $\rho$. This is because we obtain a $\frac{\epsilon}{8}$ contribution from the approximation $\tilde{\rho}$ and $\frac{\epsilon}{8}$ from statistical noise.
Thus, if $\|I/n-\hat{p}\|_{tr}\leq \frac{3\epsilon}{4}$, 
\begin{align*}
\|\sum_i\langle i|\rho|i\rangle\ketbra{i}-\frac{I}{n}\|_{tr}\leq\epsilon.
\end{align*}
by a triangle inequality, as desired. A similar argument shows that we also have $$\tr{P\left(\rho-\frac{I}{n}\right)}\geq\frac{\epsilon}{2}$$ whenever $\|I/n-\hat{p}\|_{tr}\geq \frac{3\epsilon}{4}$.
Indeed, we have:
\begin{align}
\tr{P\left(\rho-\frac{I}{n}\right)}=\tr{P\left(\hat{p}-\frac{I}{n}\right)}+\tr{P\left(\tilde{\rho}-\hat{p}\right)}+\tr{P\left(\rho-\tilde{\rho}\right)}.
\end{align}
By the definition of $P$ we have:
\begin{align}
\tr{P\left(\hat{p}-\frac{I}{n}\right)}=\|I/n-\hat{p}\|_{tr}\geq \frac{3\epsilon}{4}
\end{align}
and 
\begin{align}
\tr{P\left(\tilde{\rho}-\hat{p}\right)}+\tr{P\left(\rho-\tilde{\rho}\right)}\geq -\frac{\epsilon}{8} -\frac{\epsilon}{8}=-\frac{\epsilon}{4}.
\end{align}
This step requires a classical postprocessing time of order $\cO(n\epsilon^{-2})$. 
For implementing the second oracle, we simply measure $A \|A\|^{-1}$ directly. 
A total of $\cO(\epsilon^{-2})$ copies of $\tilde{\rho}$ suffice to determine $\tr{A\|A\|^{-1}\rho}$ up to precision $\tfrac{\epsilon}{4}$ via phase estimation \cite{Nielsen2000}. 
\end{proof}

Lemma~\ref{thm:quantumoracle} reduces the task of implementing separation oracles to the task of preparing independent copies of a fixed Gibbs state.
 There are many different proposals for preparing Gibbs states on 
 quantum computers~\cite{Chowdhury2016,Franc2017,Kastoryano2014,Poulin2009,Temme2009,Temme2009,Yung2012,Apeldoorn2017}.
 Here, we will follow the algorithm proposed in~\cite{Poulin2009}. This approach allows us to reduce the problem of preparing $\rho_H = \exp (-H)/\tr{\exp (-H)}$
to the task of simulating the Hamiltonian $H$.
 More precisely, \cite[Appendix]{Poulin2009} highlights that $ \tilde{\cO}\lb\sqrt{n}\epsilon^{-3}\rb$ invocations of a controlled $U$, where $U$ satisfies
 \begin{align*}
  \|U-e^{it_0H}\|\leq\cO(\epsilon^3) \quad \textrm{where} \quad t_0 = \pi / (4\|H \|)
 \end{align*}
suffice to produce a state that is $\frac{\epsilon}{8}$ close in trace distance to $\rho_H$. The probability of failure is constant.  We expect that a more refined analysis can lead to a better dependence on the error $\epsilon$. The methods presented in \cite{Apeldoorn2017} seem like a good starting point for such future improvements. 
Here, however, we prioritize the scaling in the problem dimension $n$ only.

By construction, the Hamiltonians we wish to simulate are all of the form $H=aA\|A\|^{-1}+bD$, where $a,b=\cO(\log(n)\epsilon^{-1})$ and $D$ is a diagonal matrix with bounded operator norm $\|D\|\leq1$. It follows from~\cite[Theorem 1]{Childs2012} that  $\tilde{\cO}\lb t(a+b) \exp (1.6\sqrt{\log( \log(n)t\epsilon^{-3}}))\rb$ separate simulations of $aA$ and $bD$ suffice to simulate $H$ for time $t$ up to an error $\epsilon^3$. Thus, we further reduce the problem of simulating $H$ to simulating $A$ and $D$ separately.

At this point, it is important to specify input models for the matrix $A$, the problem description of the \textsc{MaxQP SDP}.
We will work in the \emph{sparse oracle input model}. That is, we assume to  have access to an oracle $O_{\textrm{sparse}}$ that gives us the position of the nonzero entries. Given indices $i$ for a column of $A$ and a number $1\leq j\leq s$, where $A$ is $s$-sparse, the oracle acts as:
\begin{align*}
O_{\textrm{sparse}}\ket{i,j}=\ket{i,f(i,j)}.
\end{align*}
Here $f(i,j)$ is the index of the  $j-$th nonzero element of the $i-$th column of $A$.
Moreover, we assume that the magnitude of individual entries are accessible by means of another oracle:
\begin{align*}
O_A\ket{i,j,z}=\ket{i,j,z\oplus (A_{ij}\|A\|^{-1})},
\end{align*}
Here, the entry $\left[ A\|A\|^{-1}\right]_{ij}$ is represented by a bit string long enough to ensure the desired precision. The results of~\cite{Low2019} then highlight that it is possible to simulate $\exp(itA\|A\|^{-1})$ in time $\cO\lb \lb t\sqrt{s}\rb^{1+o(1)}\epsilon^{o(1)}\rb$. 

Let us now turn to the task of simulating diagonal Hamiltonians $D$.
Let $O_D$ be the matrix entry oracle for $D$. We suppose that it acts on $\setC^n\otimes\lb\setC^2\rb^{\otimes m}$, where $m$ is large enough to represent the diagonal
entries to desired precision in binary, as 
\begin{align}
O_D\ket{i,z}\mapsto\ket{i,z\oplus D_{ii}}. 
\end{align}
It is then  possible to simulate $H=D$ for times $t=\tilde{\cO}(\epsilon^{-1})$
with $\tilde{\cO}(1)$ queries to the oracle $O_D$ and elementary operations~\cite{Berry2007}.
Thus, efficient simulation of $\mathrm{e}^{-iDt}$ follows from an efficient implementation of the oracle $O_D$.
The latter can be achieved with a quantum RAM~\cite{Giovannetti2007}. 
We consider the quantum RAM model from \cite{Prakash2014}.
There, it is
possible to make insertions in time $\tcO\lb1\rb$. 
Thus, given a classical description of a diagonal matrix $D$, we may
update the quantum RAM in time $\tcO\lb n\rb$. After we have updated the quantum RAM, we may implement the oracle $O_D$ in time $\tcO(1)$.
Combining all these subroutines establishes the second main result of this work.

\begin{cor}
Given an $s$-sparse, symmetric $n \times n$ matrix $A$ (with appropriate oracle access) and $\epsilon >0$, we can solve the \textsc{MaxQP SDP} \eqref{eq:maxqpsdpnormal} up to an additive error
$\epsilon n\|A\|$ in time $\tcO\lb n^{1.5}\lb \sqrt{s}\rb^{1+o(1)}\epsilon^{-28+o(1)}\operatorname{exp}(1.6\sqrt{12\log(\epsilon^{-1})})\rb$ on a quantum computer. The output of the quantum algorithm consists of a real number $a$ and a diagonal matrix $D$ such that for $H=a\frac{A}{\|A\|}+D$ we have that $n\rho_H$ is at trace distance $n\epsilon$ to a feasible point of \textsc{MaxQP SDP} \eqref{eq:maxqpsdpnormal}.
\end{cor}
\begin{proof}
As we saw before, the ability to solve the relaxed  \textsc{MaxQP SDP} \eqref{equ:maxqpsdp} up to precision $\tilde{\epsilon}=\epsilon^{4}$ is sufficient to ensure an output with the properties above. 

It follows from Theorem~\ref{thm:quantumoracle} that producing $\tcO(n\tilde{\epsilon}^{-2})$ copies of Gibbs states suffices to implement the oracle. The results of~\cite{Poulin2009} then imply that each copy can be obtained with~$\tcO(\sqrt{n}\tilde{\epsilon}^{-3})$ Hamiltonian simulation steps, which, as discussed above, can each be done in time 
\begin{align*}
&\tcO\lb \lb \sqrt{s}\rb^{1+o(1)}\tilde{\epsilon}^{o(1)}\operatorname{exp}(1.6\sqrt{\log(\log(n)\tilde{\epsilon}^{-1})})\rb=\\&\tcO\lb \lb \sqrt{s}\rb^{1+o(1)}\tilde{\epsilon}^{o(1)}\operatorname{exp}(1.6\sqrt{\log(\tilde{\epsilon}^{-1})})\rb.
\end{align*}
Thus, the cost per iteration of the algorithm is 
\begin{align*}
    \tcO\lb n^{1.5}\lb \sqrt{s}\rb^{1+o(1)}\tilde{\epsilon}^{-5+o(1)}\operatorname{exp}\left[1.6\sqrt{\log(\epsilon^{-1})}\right]\rb.
\end{align*}
As the algorithm requires $\tcO(\tilde{\epsilon}^{-2})$ iterations and replacing $\tilde{\epsilon}=\cO(\epsilon^{-2})$ we obtain the claim. 

\end{proof}

\subsection{Randomized rounding} \label{sec:rounding}
As pioneered by the seminal work of Goemans and Williamson~\cite{Goemans1995}, it is possible to use randomized rounding techniques to obtain an approximate solution to the original quadratic optimization problem for certain instances~\eqref{eq:maxqp}. These solutions are in expectation within  a multiplicative factor of the value of the SDP relaxation~\eqref{equ:maxqpsdp} and the exact constant depends on the structure of the matrix $A$ . We will explore Rietz's method, as in~\cite{Alon2006}, to show that it is possible to perform the rounding on a classical computer to approximate $\|A\|_{\infty\to 1}$ with our approximately feasible solutions to \textsc{MaxQP SDP} and still obtain good approximations.

First, recall that the rounding algorithms usually work by first multiplying a random Gaussian vector by the square root of the solution. The approximate solution is then given by the signs of this random vector.
Note that both classical and quantum algorithms output a classical description of the Hamiltonian $H^\sharp$ 
associated with an approximately optimal, approximately feasible Gibbs state
$\rho^\sharp$ to~\eqref{equ:maxqpsdp}. Pseudocode for the rounding algorithm is provided in Algorithm~\ref{alg:rounding}.
The first important proof ingredient is an adaptation of \cite[Eq.~(4.1)]{Alon2006}.

\begin{lem}
Fix $v,w \in\mathbb{R}^n$ (non-zero) and let $g \in \mathbb{R}^n$ be a random vector with standard normal entries. Then,
\begin{align}
\label{equ:expectvaluerounding}
& \tfrac{\pi}{2} \mathbb{E} \left[ \sign (\langle v|g \rangle) \sign(\langle w|g \rangle) \right] \\
=& \langle \tfrac{v}{\|v \|}, \tfrac{w}{\|w \|} \rangle 
+ \mathbb{E} \left[ \left( \langle \tfrac{v}{\|v\|}| g \rangle - \sqrt{\tfrac{\pi}{2}} \sign \left( \langle \tfrac{v}{\|v \|}|g \rangle \right) \right) \left( \langle \tfrac{w}{\|w \|}| g \rangle - \sqrt{\tfrac{\pi}{2}} \sign \left( \langle \tfrac{w}{\|w \|}| g \rangle \right) \right) \right]. \nonumber
\end{align}
\end{lem}
\begin{proof}
In~\cite[Eq.~(4.1)]{Alon2006} the authors use rotation invariance to establish this identity for two unit vectors. The claim then follows from observing that the distribution of $\sign(\langle v|g \rangle)\sign(\langle w|g \rangle)$ is invariant under scaling both $v$ and $w$ by non-negative numbers. In particular, $v \mapsto v/ \| v \|$ and $w \mapsto w / \|w \|$ does not affect the distribution.
\end{proof}
The next step involves a technical continuity argument.

\begin{lem}\label{lem:multdiagonalnotchange}
Fix $\epsilon>0$ and let $\rho$ be a quantum state s.t.:
\begin{align*}
    \|\sum_i\langle i|\rho|i\rangle\ketbra{i}-I/n\|_{tr}\leq\epsilon^{4}
\end{align*}
Define the set 
$
B=\{i\in[i]:\left|\rho_{ii}-\frac{1}{n}\right|>\frac{\epsilon^2}{n}\} \label{eq:B}
$ and let $\rho_{\bar{B}}$ be the submatrix with indices in the complement $\bar{B}$ of $B$. Then, the matrix $\sigma$ with entries $\sigma_{ij} = \tfrac{\rho_{ij}}{n \sqrt{\rho_{ii} \rho_{jj}}}$ is a quantum state that obeys $\| \rho_{\bar{B}}-\sigma_{\bar{B}} \|_{tr} \leq 3 \epsilon$. 
\end{lem}
\begin{proof}
Note that
$
\sigma_{\bar{B}}=\mathcal{D}\lb \rho_{\bar{B}}\rb,
$
where $\mathcal{D}$ is the linear map given by $\mathcal{D}(X)=D_{\bar{B}}XD_{\bar{B}}$ and $D_{\bar{B}}$ is a $|\bar{B}|\times |\bar{B}|$ diagonal matrix with entries $\sqrt{n\rho_{ii}}^{-1}$ for $i\in \bar{B}$.
This implies
\begin{align*}
\|\rho_{\bar{B}}-\sigma_{\bar{B}}\|_{tr}=\|\lb \textrm{id}-\mathcal{D}\rb\lb\rho_{\bar{B}}\rb\|_{tr}\leq
\|\textrm{id}-\mathcal{D}\|_{tr\to tr}\|\rho_{\bar{B}}\|_{tr}\leq \|\textrm{id}-\mathcal{D}\|_{\infty\to\infty},
\end{align*}
because $\| \rho_{\bar{B}} \|_{tr} \leq \| \rho \|_{tr} = \mathrm{tr}(\rho)=1$.
Duality of norms and the fact that both $\textrm{id}$ and $\mathcal{D}$ are self-adjoint with respect of the Frobenius inner product $\tr{X^TY}$ implies $\|\textrm{id}-\mathcal{D}\|_{\infty\to\infty}=\|\textrm{id}-\mathcal{D}\|_{tr\to tr}$. This allows us to bound
$\|\textrm{id}-\mathcal{D}\|_{\infty\to\infty}$ instead.
By construction, we have that all the entries of $D_{\bar{B}}$ are in $1\pm\epsilon$. 
Write $D_{\bar{B}}=I+D_\epsilon$, where $D_\epsilon$ is a diagonal matrix with entries that are bounded by $\epsilon$ in absolute value. Then, 
\begin{align*}
    \textrm{id}-\mathcal{D}(X)=D_\epsilon X+X D_\epsilon+D_\epsilon XD_\epsilon
\quad \textrm{for any matrix} \quad X.
\end{align*}
Submultiplicativity of the operator norm then implies
\begin{align*}
    \|D_\epsilon X+X D_\epsilon+D_\epsilon XD_\epsilon\|_\infty\leq 2\|D_\epsilon\|_\infty\|X\|+\|D_\epsilon\|_\infty^2\|X\|_\infty\leq 3\epsilon \|X\|_\infty.
\end{align*}
and, in turn, $\|\textrm{id}-\mathcal{D}\|_{\infty\to\infty}\leq 3\epsilon$.
\end{proof}

We are now ready to prove the main stability result required for randomized rounding. 

\begin{thm} \label{prop:randomized_rounding}
Let $\rho^\sharp$ be an approximately feasible, optimal point of~\eqref{equ:maxqpsdp} with accuracy $\epsilon^4>0$ and input matrix $A'$  with
\begin{align*}
    A'=\left( \begin{array}{cc}
0& A \\
A^T &0
\end{array}
\right),
\end{align*}
where $A$ is a real $n\times n$ matrix.
Let $v_1,\ldots,v_{2n}$ be the columns of $\sqrt{\rho^\sharp}$, sample $g \in\mathbb{R}^{2n}$ with i.i.d.\ Gaussian entries and set $x_i = \sign(\langle v_i|g \rangle)$ and $y=(x_1,\ldots,x_n),z=(x_{n+1},\ldots,x_{2n})$. Then,
\begin{align*}
\tr{\rho^\sharp A}n+\cO(\epsilon n\|A\|)\geq \sum\limits_{i,j}A_{ij}\mathbb{E}(y_iz_j)\geq (4/\pi-1)\tr{\rho^\sharp A}n-\cO(\epsilon n \|A\|).
\end{align*}
\end{thm}

\begin{proof}
The upper bound follows immediately from the fact 
\textsc{MaxQP SDP} \eqref{eq:maxqpsdpnormal} relaxations (renormalized or not) provide upper bounds to the original problem \eqref{eq:maxqp}. The factors $n \| A \|$ is an artifact of the renormalization \eqref{equ:maxqpsdp}.

For the lower bound, we once more define $B=\{i\in[i]:\left|\rho_{ii}-1/2n\right|\geq\epsilon^2 /2n\} \subset \left[2n \right]$.
Plugging in $v_i$ and $v_j$ in~\eqref{equ:expectvaluerounding}, multiplying both sides by $A'_{ij}$ and summing over $i,j$ implies
\begin{align*}
&\frac{\pi}{2}\sum\limits_{i,j}A'_{ij}\mathbb{E}(x_ix_j)=2n\sum\limits_{i,j}A'_{ij}\lb\sigma_{ij}+\tau_{ij}\rb \quad \textrm{with} \quad \sigma_{ij} = \tfrac{\rho_{ij}}{2n \sqrt{\rho_{ii} \rho_{jj}}} \quad \textrm{and} \\
\tau_{ij}=&\mathbb{E} \left[ \left( \langle \tfrac{v_i}{\|v_i\|}| g \rangle - \sqrt{\tfrac{\pi}{2}} \sign \left( \langle \tfrac{v_i}{\|v_i \|}|g \rangle \right) \right) \left( \langle \tfrac{v_j}{\|v_j \|}| g \rangle - \sqrt{\tfrac{\pi}{2}} \sign \left( \langle \tfrac{v_j}{\|v_j \|}| g \rangle \right) \right) \right].
\end{align*}
Following the same proof strategy as in~\cite[Sec.~4.1]{Alon2006}, we note that the matrix $T$ defined by $\left[T \right]_{ij}=\tau_{ij}$ is a Gram matrix and, thus, psd. 
Moreover, in~\cite[Sec.~4.1]{Alon2006} the author shows that $\tau_{ii}=\frac{\pi}{2}-1$.
These two properties imply that $\lb \frac{\pi}{2}-1\rb^{-1}(2n)^{-1}T$ is a feasible point of~\eqref{equ:maxqpsdp}. 
Moreover, because of the structure of the matrix $A'$, we have that 
\begin{align}\label{equ:aboslutevalueTA}
|\tr{TA'}|\leq \lb \frac{\pi}{2}-1\rb\tr{\rho^\sharp A'}n-\cO(\epsilon n \|A\|)
\end{align} 
To see this, consider the block unitary 
\begin{align*}
    U=\left( \begin{array}{cc}
0& I \\
-I&0
\end{array}
\right).
\end{align*}
Then for any psd matrix $X$ we have that $\tr{A'UXU^\dagger}=-\tr{A'X}$ and so $\tr{A'U\rho^\sharp U^\dagger}$ provides a lower bound to the value over the approximately feasible set .
Thus,
\begin{align*}
&\frac{\pi}{2}\sum\limits_{i,j}A'_{ij}\mathbb{E}(y_iz_j)=2n\sum\limits_{i,j}A'_{ij}\lb\sigma_{ij}+\tau_{ij}\rb\geq\\
&2n\sum\limits_{i,j}A'_{ij}\sigma_{ij}-\lb \frac{\pi}{2}-1\rb\tr{\rho^\sharp A'}n-\cO(\epsilon n \|A\|)
\end{align*}
We now have to relate $\tr{\rho^\sharp A'}$ to $\tr{\sigma A'}$. To do so, we can
 argue like in Proposition~\ref{prop:stability} and see that $\tr{\sigma_{11}},\tr{\rho_{11}}=\cO(\epsilon^2)$ (these correspond to the $|B| \times |B|$ psd submatrices with entries in $B$ only). As both $\sigma$ and $\rho$ are states,  we conclude 
\begin{align*}
\|\rho_{12}\|_{tr},\|\sigma_{12}\|_{tr}=\cO(\epsilon)
\end{align*}
by reusing the analysis provided in the proof of Proposition~\ref{prop:stability}.
Thus, it follows from H\"olders inequality and Lemma~\ref{lem:multdiagonalnotchange} that
\begin{align*}
\tr{A'(\rho-\sigma)}=&\tr{A'(\rho_{22}-\sigma_{22})}+\tr{A'(\rho_{11}-2\rho_{12}-\sigma_{11}-2\sigma_{12})}\\
=&\|A\|\lb \|\sigma_{22}-\rho_{22}\|_{tr}+\| \rho_{11}\|_{tr}+2\|\rho_{12}\|_{tr}+2\| \sigma_{11}\|_{tr}+2\|\sigma_{12}\|_{tr}\rb=\cO(\|A\|\epsilon),
\end{align*}
from which the claim follows.
\end{proof}
Proposition~\ref{prop:randomized_rounding} highlights that performing the rounding with approximate solutions to the \textsc{MaxQP SDP} \eqref{equ:maxqpsdp} still ensures a good approximate solution in expectation for the $\|A\|_{\infty\to 1}$ norm. 
In the case of matrices $A$ that are psd it is possible to improve the constant in the rounding and we do not to resort to lifting the problem to a matrix with double the dimension:
\begin{cor} \label{prop:randomized_rounding_psd}
Let $\rho^\sharp$ be an approximately feasible, optimal point of~\eqref{equ:maxqpsdp} with accuracy $\epsilon^4>0$ and psd input matrix $A$.
Let $v_1,\ldots,v_{n}$ be the columns of $\sqrt{\rho^\sharp}$, sample $g \in\mathbb{R}^{n}$ with i.i.d.\ Gaussian entries and set $x_i = \sign(\langle v_i| g \rangle)$. Then,
\begin{align*}
\tr{\rho^\sharp A}n+\cO(\epsilon n\|A\|)\geq \sum\limits_{i,j}A_{ij}\mathbb{E}(x_ix_j)\geq (2/ \pi)\tr{\rho^\sharp A}n-\cO(\epsilon n \|A\|).
\end{align*}
\end{cor}
\begin{proof}
The proof follows by following the same proof as above but noting that we may use the estimate $\tr{TA}\geq 0$ instead of~\eqref{equ:aboslutevalueTA}, as both $A$ and $T$ are psd. Optimality of the constant was shown in~\cite{Alon2006}.
\end{proof}
As Alon and Naor~\cite{Alon2006} also show that for psd matrices $A$ we have 
\begin{align*}
\|A\|_{\infty\to1}=\max\limits_{x\in\{\pm 1\}^n}\langle x|A|x\rangle,
\end{align*}
i.e. we may restrict to the same vector on the left and right, it follows that Corollary~\ref{prop:randomized_rounding_psd} gives almost optimal rounding guarantees. These two statements certify that, as longs as $\|A\|_{\infty\to1}=\Theta(n\|A\|)$, performing the rounding with our approximately feasible solutions gives rise to approximations of the $\infty\to 1$ norm that are almost as good the strictly feasible solutions.

But computing $\sqrt{\rho^\sharp}g= \exp (-H/2)g/\sqrt{\tr{\exp(-H)}}$ directly still remains expensive because of matrix exponentiation. 
We will surpass this bottleneck by truncating the Taylor series of the matrix exponential in a fashion similar to Lemma~\ref{lem:errortruncation}. The following standard anti-concentration result for Gaussian random variables will be essential for this argument.

\begin{fact}\label{lem:gaussiananti}
Let $X$ be a $\mathcal{N}(0,\sigma^2)$ random variable. Then
$
\mathbb{P}(|X|\leq\sigma\epsilon)=\cO(\epsilon).
$
\end{fact}

\begin{lem}\label{lem:truncatingworks}
Let $\rho^\sharp$ with associated Hamiltonian $H^\sharp$ be an approximately optimal solution  to the \textsc{MaxQP SDP} \eqref{equ:maxqpsdp} with $\|H^\sharp \|=\cO (\log (n)/\epsilon)$. Set
$S_l = \sum_{k=0}^l \tfrac{1}{k!} (-H^\sharp/2)^k$ with $l= \cO (\log (n) /\epsilon)$.
Then, a random vector $g \in \mathbb{R}^n$ with standard normal entries obeys
\begin{align*}
\sign\left[\lb e^{H^\sharp/2}g\rb_i\right]=\sign\left[\lb S_l g\rb_i\right] \quad \textrm{for all} \quad i \in \left[n\right] \quad \textrm{such that} \quad \left| \rho_{ii}^\sharp - \tfrac{1}{n} \right| < \tfrac{\epsilon}{n}
\end{align*}
 with probability at least $1-\cO(\epsilon^{-1})$.
\end{lem}

Note that the design of Algorithm~\ref{alg:HU} ensures that optimal Hamiltonians always obey $\| H^\sharp \| = \cO (\log (n)/\epsilon)$.

\begin{proof}
Define $h = \exp (-H^\sharp/2) g$ and note that this is a Gaussian random vector with covariance matrix $\exp (-H^\sharp)$. Let $B= \left\{i:\; \left| \rho_{ii} - 1/n \right|>\frac{\epsilon}{n} \right\} \subset \left[n\right]$ denote the set of indices for which $\rho_{ii}$ deviates substantially from $1/n$. 
Then, every entry of $h$ that is not contained in this index set obeys
\begin{equation*}
\left[ h \right]_i = \left[ \exp (-H/2) g \right]_i \sim \mathcal{N} \left( 0, \tfrac{c}{n} \mathrm{tr}(\exp (-H)) \right) \quad \textrm{with} \quad c \in (1-\epsilon,1+\epsilon).
\end{equation*}
The assumption
 $\| H^\sharp \| = \cO (\log (n)/\epsilon)$ ensures $\mathrm{tr}(\exp (-H^\sharp) )/n \geq n^{-c'/\epsilon-1}$ for some constant $c'$. We can combine this with Fact~\ref{lem:gaussiananti} (Gaussian anti-concentration) to conclude
\begin{equation*}
\mathbb{P} \left[ | \left[h\right]_i | \leq n^{-2-c'/(2\epsilon)} \right]= \cO (1/n^2) \quad \textrm{for all} \quad i \in \bar{B} = \left[n\right] \setminus B.
\end{equation*}
A union bound then asserts
\begin{equation*}
\mathbb{P} \left[ \exists i \in \bar{B}:\; | \left[h\right]_i| \leq n^{-2-c'/\epsilon} \right] = \cO (1/n).
\end{equation*}
Moreover, it follows from standard concentration arguments that 
\begin{align*}
\mathbb{P}\left[ n-n^{\frac{1}{4}}\leq\|g\|^2\leq n+n^{\frac{1}{4}}\right] \geq1- 2e^{-\sqrt{n}/8}.
\end{align*}
Thus, with probability at least $1-\cO(n^{-1})$, we have that $\|g\|^2\leq n+n^{\frac{1}{4}}$ and $| \left[h\right]_i | \geq n^{-2-c'/\epsilon}$ for every entry $i \in \bar{B}$. 
Following the same proof strategy as in Lemma~\ref{lem:errortruncation}, it is easy to see that picking $l=\cO(\epsilon^{-1}\log(n))$ suffices to ensure that
\begin{align*}
\|S_l-\exp (-H/2)\|\leq n^{-4-\frac{c'}{2\epsilon}}
\end{align*}
Conditioning on the events emphasized above, implies
\begin{align*}
\max\limits_{i\in[n]}\left|\left[\lb \exp(-H/2)-S_l\rb g\right]_i\right|\leq\|\lb \exp (-H/2)-S_k\rb g\|\leq
\| \exp (-H/2)-S_k\|\|g\|\leq n^{-4-\frac{c'}{2\epsilon}}\|g\|.
\end{align*}
This in turn ensures
$
\max\limits_{i\in \bar{B}}\left|\left[\lb \exp (-H/2)-S_l\rb g\right]_i\right|\leq n^{-3-\frac{c'}{2\epsilon}}
$, which then gives
\begin{align*}
\sign\left( \left[h\right]_i \right) = \sign\left(\left[ (\exp (-H/2) g) \right]_i \right)=
\sign\left( \left[ S_k g \right]_i \right) \quad \textrm{for all} \quad i \in \bar{B},
\end{align*}
because conditioning ensures $\left|\left[ \exp (-H/2)g\right]_i\right|\geq n^{-2-\frac{c'}{2\epsilon}}$.
\end{proof}

Combining the statements we just proved we conclude that:

\begin{prop}[Restatement of Proposition~\ref{prop:rounding_main}]
Let $\epsilon>0$ and $A$ a real, psd matrix be given. Moreover, let $H$ be the solution Hamiltonian to the relaxed \textsc{MaxQP SDP}~\eqref{equ:maxqpsdp} with error parameter $\epsilon^4$ and $\alpha^*$ its value. Then, with probability at least $1-n^{-1}$, the output $x$ of Algorithm~\ref{alg:rounding} satisfies:
\begin{align}\label{equ:upperandlowersign}
n\|A\|(\alpha^*+\cO(\epsilon))\geq\mathbb{E}\big[\sum\limits_{ij}A_{ij}x_ix_j\big]\geq \frac{2}{\pi}n\|A\|(\alpha^*-\cO(\epsilon)),
\end{align}
\end{prop}
\begin{proof}
It follows from Lemma~\ref{lem:truncatingworks} that the output of Algorithm~\ref{alg:rounding} will only differ from the vector obtained by performing the rounding with the approximate solution on a set of size $\cO(n\epsilon^2)$ with probability at least $1-n^{-1}$. This is because, as argued before, by picking $\epsilon^4$ we have at most $\cO(\epsilon^2n)$ diagonal entries that do not satisfy $|\rho_{ii}-1/n|\leq\epsilon/n$.
We will now argue that sign vectors that differ at $\cO(n\epsilon^2)$ position can differ in value  by at most $\cO(\epsilon n\|A\|)$.  Let $x$ be the vector obtained by the ideal rounding and $x'$ the one with the truncated Taylor series. Then there exists a vector $e$ with at most $\cO(n\epsilon^2)$ nonzero entries bounded by $2$ such that $x=x+e$ by our assumption. By Cauchy-Schwarz:
\begin{align*}
|\langle x|A|x\rangle-\langle x'|A|x'\rangle|\leq |\langle e|A|x\rangle|+|\langle x|A|e\rangle|+|\langle e|A|e\rangle|\leq \|A\|\lb2 \|x\|\|e\|+\|e\|^2\rb.
\end{align*}
Now, as $x$ is a binary vector, $\|x\|=\sqrt{n}$ and, as $e$ has at most $\cO(\epsilon^2 n)$ nonzero entries, it follows that $\|e\|=\cO(\epsilon\sqrt{n})$ and we conclude
\begin{align*}
|\langle x|A|x\rangle-\langle x'|A|x'\rangle|=\cO(\epsilon n \| A \|)
\end{align*}
As Theorem~\ref{prop:randomized_rounding} asserts that performing the rounding with the approximate solution is enough to produce a sign vector that satisfies~\eqref{equ:upperandlowersign} in expectation, this yields the claim.  
\end{proof}
The analogous claim, i.e. that truncating still gives rise to good solutions, clearly also holds in the setting of Proposition~\ref{prop:randomized_rounding}.

Thus, we conclude that the rounding can be performed in time $\tcO(ns)$ on a classical computer, as multiplying a vector with $H$ takes time $\tcO(ns)$ and we only need to perform these operations for a total number of steps that is logarithmic in the problem dimension $n$ (but polynomial in inverse accuracy $1/\epsilon$).
 As $ns\leq n^{1.5}\sqrt{s}$ for $s\leq n$, we conclude that the cost of solving the relaxed \textsc{MaxQP SDP} \eqref{equ:maxqpsdp} dominates the cost of rounding.

\section{Conclusion and Outlook}

By adapting ideas from \cite{Tsuda2005,Hazan2006,Steurer2015,Brandao2017b},
we have provided a general meta-algorithm for approximately solving 
convex feasibility problems with psd constraints. \emph{Hamiltonian Updates} is an iterative procedure based on a simple change of variables: represent a trace-normalized, positive semidefinite matrix as $X=\exp(-H)/\tr{\exp (-H)}$. At each step, infeasible directions are penalized in the matrix exponent until an approximately feasible point is reached.
This procedure can be equipped with rigorous convergence guarantees and lends itself to quantum improvements: $X = \exp (-H)\tr{\exp (-H)}$ is a \emph{Gibbs state} and $H$ is the associated \emph{Hamiltonian}. Quantum architectures can produce certain Gibbs states very efficiently. 

We have demonstrated the viability of this approach by considering semidefinite programming relaxations of quadratic problems with binary constraints (\textsc{MaxQP SDP}) \eqref{eq:maxqpsdpnormal}. 
The motivation for considering this practically important problem class was two-fold: (i) \textsc{MaxQP SDP}s have received considerable attention in the (classical) computer science community. Powerful meta-algorithms, like matrix multiplicative weights \cite{Arora2007}, have been designed to solve these SDPs very quickly.
(ii) So far, quantum SDP solvers \cite{Brandao2017a,Apeldoorn2017,Apeldoorn2018,Brandao2017b,Kerenidis2018} have failed to provide speedups for \textsc{MaxQP SDP}s. The quantum runtime associated with these solvers depends on problem-specific parameters that scale particularly poorly for \textsc{MaxQP SDP}s. Moreover, the notions of approximate feasibility championed in these other works are too loose for this class of problem.

The  framework developed in this paper has allowed us to address these points.
Firstly, we shown that a classical implementation of Hamiltonian Updates already improves upon the best existing results. A runtime of $\tcO (n^2 s)$ suffices to find an approximately optimal solution.
Secondly, we have showed that quantum computers do offer additional speedups. A quantum runtime of $\tcO (n^{1.5}s^{0.5+o(1)})$ is sufficient. 
We emphasize that this is the first quantum speedup for %
\textsc{MaxQP SDP} relaxations.
Subsequently, we have devised a classical randomized rounding procedure that converts both quantum and classical solutions into close to optimal solutions of the original quadratic problem.

We note in passing that our algorithm is very robust, in the sense that it only requires the preparation of Gibbs states up to a precision $\epsilon$ that can be taken to be constant in the number of qubits. This requirement is combined with other simple tasks like computational basis measurements and the ability to estimate the expectation value of the target matrix on states. Although the subroutines used in this work to perform these tasks certainly require nontrivial quantum circuits, it would be interesting to identify classes of target matrices $A$ for which preparing the corresponding Gibbs state and estimating the expectation values is feasible on near-term devices.

We believe that the framework presented here lends itself to further applications. 

One concrete application of Hamiltonian Updates, in particular the idea to treat constraints as the statistics of measurements,
is \emph{quantum state tomography}, see e.g.\ \cite{Gross2013} and references therein.
Sample-optimal tomography protocols have revealed that classical postprocessing is the main bottleneck for reconstructing density matrices \cite{Flammia2012,Wright2016,kueng2020low,Haah2017,Guta2018}. 
A classical implementation of Hamiltonian Updates allows for optimizing postprocessing costs at the expense  of  a  worse  dependence  on  accuracy \cite{brandao2020tomography}.
Further improvements are possible by executing the algorithm on a quantum computer, giving a quantum speedup for quantum state tomography.

Another promising and practically relevant application is \emph{binary matrix factorization}. 
A recent line of works \cite{Kueng2019a,Kueng2019b} reduces this problem to a sequence of SDPs. 
Importantly, each SDP corresponds to a \textsc{MAXQP SDP} \eqref{eq:maxqpsdpnormal} with a random rank-one objective $A = |a \rangle \! \langle a|$ and an additional affine constraint $\tr{P X}=n$. Here, $P$ is a fixed low-rank orthoprojector. 
This application, however, is likely going to be more demanding in terms of approximation accuracy. Hence, improving the runtime scaling in inverse accuracy will constitute an important first step that is of independent interest.

\section{Acknowledgments}
We would like to thank Aram Harrow for inspiring discussions. 
Our gratitude extends, in particular, to 
Ronald de Wolf, Andr\'as Gily\'en and Joran van Apeldoorn who provided valuable feedback regarding an earlier version of this draft.
Finally, we would like to thank the anonymous reviewers for in-depth
comments and suggestions.
D.S.F. would like to thank the hospitality of Caltech's Institute for Quantum Information and Matter, where the 
main ideas in this paper were conceived during a visit.
F.G.L.S.B and R.K. acknowledge  funding provided by the Institute for Quantum Information and Matter, an NSF Physics Frontiers Center (NSF Grant PHY-1733907), as well as financial support from Samsung. R.K's work is also supported by the Office of Naval Research (Award N00014-17-1-2146) and the Army Research Office (Award W911NF121054).
D.S.F. acknowledges financial support from VILLUM FONDEN via the QMATH Centre of Excellence (Grant no. 10059), the graduate program TopMath of the Elite Network of Bavaria, the 
TopMath Graduate Center of TUM Graduate School at Technische Universit\"{a}t M\"{u}nchen and by the 
Technische Universit\"at M\"unchen – Institute for Advanced Study,
funded by the German Excellence Initiative and the European Union Seventh Framework
Programme under grant agreement no. 291763.

\bibliographystyle{abbrvnat}
\bibliography{cutnorm}

\appendix

\section{Norms of random matrices}\label{app:randommatrices}

There is an interesting discrepancy in the error scaling between the methods presented here and existing ones by Arora et al.~\cite{Arora2005}: $\| A \|_{\ell_1} $ (existing work) vs $n \| A \|$ (here). The following fundamental relations relate these norms \cite{Nikiforov2009}:
\begin{align*}
\|A\|_{\infty\to 1}\leq n\|A\|,\quad \|A\|_{\infty\to 1}\leq \|A\|_{\ell_1},\quad \|A\|\leq \|A\|_{\ell_1}\leq n\sqrt{\textrm{rank}(A)}\|A\|.
\end{align*}
All inequalities are tight up to constants.
The above inequalities highlight that it is a priori not clear what the correct scaling for errors approximating the cut norm should be. The goal of this section will be to show that for random matrices $A$ with independent, standardized entries that have bounded fourth moment
$n \| A \|$ reproduces the correct error behavior, while $\| A \|_{\ell_1}$ does not.

\begin{prop}[Cut norm of random matrices]
Let $A$ be a $n\times n$ random matrix whose entries are sampled independently from a real-valued distribution $\alpha$ that obeys $\mathbb{E} \left[\alpha \right]=0$, $\mathbb{E} \left[ \alpha^2 \right]=1$ and $\mathbb{E} \left[ \alpha^4 \right] = \cO(1)$. Then,
\begin{align*}
\mathbb{E}\left[ \|A\|_{\ell_1}\right]=\Theta(n^2),\quad\mathbb{E}\left[ \|A\|_{\infty\to1}\right]=\Theta(n^{1.5}),\quad\mathbb{E} \left[ \|A\|\right]=\cO\lb\sqrt{n}\rb.
\end{align*}

\end{prop}
\begin{proof}
We refer to Latala's work for the third claim \cite{Latala2005}.
A key ingredient for establishing the second claim is 
~\cite[Corollary 3.10]{Gittens2013}:
\begin{align*}
\tfrac{1}{ \sqrt{2}} \mathbb{E}\left(\|A\|_{\mathrm{col}}\right) \leq \mathbb{E}\lb\|A\|_{\infty \rightarrow 1}\rb \leq 4 \mathbb{E}\left(\|A\|_{\mathrm{col}}\right),
\end{align*}
where $\|A \|_{\mathrm{col}} = \sum_i \sqrt{ \sum_j \left[A\right]_{ij}^2 }$ is the sum of the Euclidean norms of the columns of $A$.
Now, note that the entries of $A$ are i.i.d.\ copies of the random variable $\alpha$. 
In turn, the expected column norm of $A$ is just $n$ times the expected Euclidean norm of the random vector $a=(a_1,\ldots,a_n)^T$, where each $a_i$ is an independent copy of $\alpha$. Jensen's inequality then asserts
\begin{equation*}
\mathbb{E}\left[ \| a \|_{2}\right] \leq \left( \mathbb{E} \left[ \sum_{i=1}^n a_i^2 \right] \right)^{1/2} = \sqrt{n \mathbb{E} \left[ \alpha^2 \right]} = \sqrt{n},
\end{equation*}
while a matching lower bound follows from $\sqrt{x} \geq \tfrac{1}{2}(1+x-(x-1)^2)$. Indeed, define $y = \| a \|_{2}^2/n= \tfrac{1}{n} \sum_{i=1}^n a_i^2$ and note that this new random variable obeys $\mathbb{E}[y]=1$ and $\mathbb{E}[(y-1)^2]= \mathcal{O}(1/n)$ by assumption. This ensures a matching lower bound:
\begin{align*}
\mathbb{E} \left[ \| a \|_{2} \right] =\sqrt{n} \mathbb{E} \left[\sqrt{y} \right] \geq \tfrac{\sqrt{n}}{2} \left( 1+ \mathbb{E} \left[ y \right] - \mathbb{E}\left[(y-1)^2\right]\right) = \Omega(\sqrt{n}),
\end{align*}
This ensures $\mathbb{E} \left[ \| A \|_{\infty \to 1} \right] = n \mathbb{E} \left[ \| a \|_{2} \right] = \Theta (n^{3/2})$ and establishes the second claim.

The first claim follows from the fact that the fourth-moment bound $\mathbb{E} \left[ \alpha^4 \right] = \mathcal{O}(1)$ demands $\mathbb{E} \left[ | \alpha | \right] = \Theta (1)$. Combine this with i.i.d.\ entries of the random matrix $A$ to conclude
\begin{equation*}
\mathbb{E} \left[ \| A \|_{\ell_1} \right] = n^2 \mathbb{E} \left[ | \alpha| \right] = \Theta (n^2).
\end{equation*}
\end{proof}

In the case of random matrices with Gaussian entries, such as in the case of the SK-model, we have also have exponential concentration around these expectation values, as shown in~\cite[Theorem 1.2]{Panchenko2013}.

Another family of random matrices for which we expect that $n \| \cdot \|$ provides the correct error scaling for cut norms are matrices of the form $B = A^* A$, where 
 $A$ again has i.i.d. entries of mean $0$ and unit variance. Indeed, in~\cite{Rebrova2018} the authors show that
\begin{align*}
\mathbb{E}\lb \|A\|_{\infty\to2}\rb\leq \cO\lb \sqrt{n}\mathbb{E}\lb \|A\|_{2\to\infty}\rb\rb.
\end{align*}
with high probability. 
One can combine these recent results with more standard relations, like
$\|A\|_{2\to\infty}^2=\|B\|_{1\to\infty}$, $\|A\|_{2\to\infty}\leq \|A\|$ and $\|B\|=\|A\|^2$. 
This asserts $ \mathbb{E} \left[\| B \|_{1 \to \infty} \right] \leq n \mathbb{E} \left[ \| B \| \right]=\mathcal{O}(n^2)$, while $\mathbb{E} \left[ \| B \|_{\ell_1} \right] = \Omega (n^{2.5})$. 

\section{Random instances of the \textsc{MaxQP SDP}}\label{app:randominstances}
The following random instances of the \textsc{MaxQP SDP} have received significant attention in recent literature~\cite{montanari_semidefinite_2016,Kunisky2019,Javanmard_2016}: we define the Gaussian orthogonal ensemble (GOE) to be the random matrix distribution over symmetric $n\times n$ matrices $A$ with
i.i.d. normal entries $A_{ii}\sim \mathcal{N}(0,2/n)$ on the diagonal and $A_{ij}\sim \mathcal{N}(0,1/n)$ for $i<j$.
For a parameter $\lambda \geq 0$, we also define the deformed GOE as $B(\lambda)= (\lambda/n)\mathbf{1}\mathbf{1}^T+A$,
where $\mathbf{1}=(1,\ldots,1)^T \in \mathbb{R}^n$ is the vector of ones and $A$ is sampled from the GOE.

This random matrix model is intimately connected to the Sherrington-Kirkpatric model and determining the MaxCut of random graphs. We will instantiate the notation from~\cite{montanari_semidefinite_2016} and write \textsc{MaxQP}$(B(\lambda))$ for the optimal value of the \textsc{MaxQP SDP} with input $B(\lambda)$.
In~\cite{montanari_semidefinite_2016}, Montanari and Sen showed the following interesting result.

\begin{thm}[Theorem 5 of~\cite{montanari_semidefinite_2016}]\label{thm:phase_trans_sdp}
Fix $\lambda \geq 0$ and sample $B(\lambda)$ from the associated deformed GOE.
\begin{enumerate}
    \item If $0 \leq \lambda \leq 1$, then for any $\epsilon>0$, we have \textsc{MaxQP}$(B(\lambda))/n\in[2-\epsilon,2+\epsilon]$ with probability converging to $1$ as $n\to\infty$.
    \item Else if $\lambda>1$, then there exists a constant $\Delta(\lambda)>0$ such that \textsc{MaxQP}$(B(\lambda))/n\geq 2+\Delta(\lambda)$ with probability converging to $1$ as $n\to\infty$.
\end{enumerate}
\end{thm}

This seemingly abstract theorem has profound implications to our work. To appreciate them,  it is worth noting that 
for $0 \leq \lambda \leq 1$, it is known that the maximal eigenvalue of $B(\lambda)$ is also contained in the interval $[2-\epsilon,2+\epsilon]$ with high probability~\cite{Knowles_2013}. Thus, the optimal value of \textsc{MaxQP} is comparable in size to the re-scaled largest eigenvalue $n\|B (\lambda) \|$. Moreover, it also follows that for these instances $\|B(\lambda)\|_{\ell_1}=\Omega($\textsc{MaxQP}$(B(\lambda)) \sqrt{n}$ in expectation.

On the other hand, if $\lambda>1$, the largest eigenvalue of the matrix and  $B(\lambda)$ is given by $\lambda+\lambda^{-1}$~\cite{Knowles_2013}. We see that both the largest eigenvalue and the value of \textsc{MaxQP}$(B(\lambda))$ go through a phase transition at $\lambda=1$.

Let us now focus on the case $\lambda<1$. Note that the dual of the \textsc{MaxQP SDP} with target matrix $A$ is given by optimizing over $y\in\R^{n+1}$ as follows:
\begin{align}
\textrm{minimize} \quad& n y_0+\sum_{i=1}^n y_i & (\textsc{dual MaxQP SDP}) \label{eq:dualmaxqpsdpnormal} \\
\quad &\textrm{subject to} \quad y_0I+\mathrm{diag}(y)\geq A,\; y_i\geq 0   \nonumber,
\end{align}
where $\mathrm{diag}(y)=\mathrm{diag}(y_1,\ldots,y_n)$ denotes the diagonal matrix with entries $y_i$ for $1 \leq i \leq n$.
The additional dual variable $y_0$ arises from also incorporating the redundant constraint $\tr{X}\leq n$ in the associated primal \textsc{SDP}~\eqref{equ:maxqpsdp}. This choice is motivated by the observation that previous quantum \textsc{SDP} solvers~\cite{Brandao2017b,Apeldoorn2017,Apeldoorn2018} actually output approximately optimal solutions for the dual SDP with this redundant constraint. Note once again that~\cite{Apeldoorn2018} also proposes a primal-only algorithm that bears strong similarities to ours. However, their version of the algorithm does not readily admit a quantum speedup because of the way they implement the diagonal constraints. We refer to Sec.~\ref{sub:related} for more details.

Theorem~\ref{thm:phase_trans_sdp} implies that we can always find a trivial feasible point that is approximately optimal and sparse for Eq.~\eqref{eq:dualmaxqpsdpnormal}. Indeed, for any $\epsilon>0$ fixed, $\gamma=(2+\epsilon)e_0$ will be feasible with probability $1$ in the limit $n\to\infty$. We conclude that for $\epsilon>0$ fixed and in the regime $\lambda<1$, solving the dual problem is trivial. So, existing solvers that output a feasible, approximately optimal dual solution~\cite{Brandao2017b,Apeldoorn2017,Apeldoorn2018} for $\epsilon$ of constant order are of little practical interest for the problem at hand.
In stark contrast, feasible and approximately optimal solutions of the primal problem are still relevant, because they can be used to perform the rounding.
It is also important to note that the proof of~\cite[Theorem 5]{montanari_semidefinite_2016} is constructive.
Indeed, let $P_{\delta}$ denote the the projector onto the range of the best rank-$\delta n$ approximation of $B(\lambda)$, that is, the subspace spanned by the the eigenvectors corresponding to the largest $\delta n$ eigenvalues of $B(\lambda)$. Moreover, let $D$ with $(D)_{ii}=(P_\delta)_{ii}$ be the restriction of this projector to the main diagonal.
By construction, $(D)_{ii} =(P_{\delta})_{ii}>0$ almost surely for all $1 \leq i \leq n$. And, in turn,  $X=D^{-\frac{1}{2}}P_{\delta}D^{-\frac{1}{2}}$ must be a feasible point of the primal \textsc{MaxQP SDP}~\eqref{eq:maxqpsdpnormal}.
Montanari and Sen then show that $\tr{B(\lambda)X}\geq 2-\epsilon$ for some $\epsilon=\Omega(\delta)$. This establishes the first part of Thm.~\ref{thm:phase_trans_sdp}. In summary, diagonalizing $B(\lambda)$ and computing $X$ is sufficient to obtain an approximately optimal primal solution. Suppressing the error dependence on $\epsilon$, diagonalizing $B(\lambda)$ takes $\cO(n^{\omega})$ time, while our classical algorithm to solve \textsc{MaxQP SDP} takes the same up to polylogarithmic factors. The quantum runtime, however, is of order $\tilde{\cO}(n^{2+o(1)})$ only. Thus, we see that for $\epsilon=\Theta(1)$ we obtain a quantum speedup as soon as the exponent of matrix multiplication obeys $\omega>2$ (which is widely believed).

Let us now discuss the regime where $\lambda>1$. To the best of our knowledge, 
the limit value of \textsc{MaxQP}$(B(\lambda))/n$ has not been identified yet.
Ref.~\cite{montanari_semidefinite_2016}, however, shows that it must be strictly larger than $2$ by constructing a sequence of feasible points that continues to saturate such a lower bound.
However, in contrast to before ($\lambda <1$), it is not clear that this sequence of feasible points is approximately optimal. In fact, it is not even known if \textsc{MaxQP}$(B(\lambda))/n<\lambda+\lambda^{-1}$.'
But numerical evidence in favor of this behavior is provided in~\cite{Javanmard_2016}. That is, 
the optimal value of the SDP is strictly smaller than the trivial eigenvalue upper bound. Thus, if it is indeed the case that \textsc{MaxQP}$(B(\lambda))/n+\mu<\lambda+\lambda^{-1}$ for some $\mu>0$ as $n\to\infty$, then the dual SDP must be nontrivial to solve. %
Still, a direct solution of the primal problem is arguably more relevant, because it can be used to perform randomized rounding.
Nevertheless, we will now argue that previous quantum solvers~\cite{Brandao2017b,Apeldoorn2017,Apeldoorn2018} do not give rise to a speedup for the dual problem assuming that \textsc{MaxQP}$(B(\lambda))/n+\mu<\lambda+\lambda^{-1}$.

Before we move on, we once more emphasize that our algorithm considers the primal problem only. This is in stark contrast to existing quantum SDP solvers that address both primal and dual problems.
This fully primal approach has the advantage that the runtime of the algorithm does not depend on problem-specific parameters like the $\|\cdot\|_{\ell_1}$ norm of approximately optimal dual solutions, as mentioned in Sec.~\ref{sub:related}. We will now show that this becomes a real advantage for solving the \textsc{MaxQP SDP} for $B(\lambda)$ in the regime $\lambda >1$. 

\begin{prop}\label{prop:nosparsedualsolution}
Let $y^\sharp$ be a $\delta$-approximately optimal solution to \textsc{dual MaxQP SDP}~\eqref{eq:dualmaxqpsdpnormal} for $B(\lambda)$ with $\lambda>1$. Assume, moreover, that there exists a $\mu >0$ such that that 
\begin{align*}
    \textsc{MaxQP}(B(\lambda))/n+\mu\leq \|B(\lambda)\|
\end{align*}
with high probability. Then, with high probablity,  
every $\delta$-optimal dual solution satisfies
\begin{align}\label{equ:sparsity_dual}
\left\|y^\sharp\right\|_{\ell_1}=\Omega(n) \quad \text{as long as $\delta < \mu$.}
\end{align}
\end{prop}
\begin{proof}
Let $x^*$ be the value achieved by $y^\sharp$ and set $\eta=\mu-\delta>0$. 
If we condition on the event $x^*+\eta\leq \|B(\lambda)\|$, the feasibility constraint in the dual \textsc{SDP}~\eqref{eq:dualmaxqpsdpnormal} enforces $y_0+\eta\leq \|B(\lambda)\|$. To see this, note that the value of the dual \textsc{SDP} is clearly monotonically increasing on the other entries, which are all positive. We will now show that in order for the matrix inequality
\begin{align}\label{equ:feasible_dual}
    y_0I+\mathrm{diag}(y^\sharp)\geq B(\lambda)
\end{align}
to hold, then the vector $y^\sharp\in\R^n$ must have large $\ell_1$ norm. In order to do this, we will resort to an approximate leading eigenvector construction by~\cite{montanari_semidefinite_2016}. This construction will have the desirable property that it is not too ``spiky''. In turn, this approximate leading eigenvector will have a small overlap with each entry of the diagonal matrix $\mathrm{diag}(y)$.

We will make extensive use of the results of~\cite{montanari_semidefinite_2016}, so we will also follow their notation and normalizations for this proof. Define $u_1$ to be the eigenvector corresponding to the largest eigenvalue of $B(\lambda)/n$. Moreover, define the ``capping'' function $R(x)$ as
\begin{align*}
R(x)=
\begin{cases}
-1,\quad\textrm{if } x<-1\\
x,\quad\textrm{if } -1\leq x\leq1\\
1,\quad\textrm{if }x>1\\
\end{cases}
\end{align*}
For some $\epsilon >0$, Montanari and Sen then define the vector  $\varphi$ componentwise as $\varphi_i=R(\epsilon\sqrt{n}u_{1,i})$. In \cite[Lemma G.2]{montanari_semidefinite_2016}, they then establish
\begin{align}\label{equ:approximate_eigenvalue}
    \left|\frac{1}{n}\tr{\ketbra{\varphi}B(\lambda)}-\epsilon^2\|B(\lambda)\|\right|=\mathcal{O}(\epsilon^4) \quad \text{with high probability.}
\end{align}
On top of that, in Eq. (163) they show that:
\begin{align}\label{equ:norm_inequality_phi}
    \frac{1}{n}\left\|\epsilon\sqrt{n}u_{1}-\varphi\right\|_2^2=\cO(\epsilon^6).
\end{align}
We can now use the vector $\varphi$ to probe positive semidefiniteness in Eq.~\eqref{equ:feasible_dual}:
\begin{align}\label{equ:matrix_ineq_phi}
 \tr{\ketbra{\varphi}B(\lambda)} \leq   \tr{\ketbra{\varphi}( y_0I+\mathrm{diag}(y)}.
\end{align}
Let us start by estimating the right-hand side of this scalar inequality. Combining Eq.~\eqref{equ:norm_inequality_phi} with a reverse triangle inequality yields
\begin{align*}
    \tr{\ketbra{\varphi} y_0I}\leq y_0n(\epsilon^2+\epsilon^4).
\end{align*}
Furthermore, by construction, the entries of $\varphi$ squared are at most $1$. Thus,
\begin{align*}
    \tr{\ketbra{\varphi} \mathrm{diag}(y)} = \sum\limits_{i=1}^n |\varphi_i|^2 y_i \leq \sum\limits_{i=1}^ny_i
\end{align*}
and we can combine both arguments to obtain an upper bound on the r.h.s.\ of Eq.~\eqref{equ:matrix_ineq_phi}.
We can also lower-bound the l.h.s.
Eq.~\eqref{equ:approximate_eigenvalue} asserts
\begin{align*}
    \tr{\ketbra{\varphi}B(\lambda)}\geq (\epsilon^2+C\epsilon^4)n,
\end{align*}
for some constant universal $C>0$. Putting these inequalities together we conclude
\begin{align*}
 y_0n(\epsilon^2+\epsilon^4)+\sum\limits_{i=1}^ny_i\geq \epsilon^2\|B(\lambda)\|n-C\epsilon^4n.
\end{align*}
Dividing the inequality by $\epsilon^2 n$ and using the fact that our conditioning guarantees $y_0+\eta\leq \|B(\lambda)\|$ we conclude
\begin{align*}
    (\|B(\lambda)\|-\eta)(1+\epsilon^2)+(n\epsilon^2)^{-1}\sum\limits_{i=1}^ny_i\geq \|B(\lambda)\|-C\epsilon^4.
\end{align*}
Rearranging the terms produces
\begin{align*}
    n^{-1}\sum\limits_{i=1}^ny_i\geq \epsilon^2(\eta-\epsilon^2\|B(\lambda)\|-C\epsilon^4).
\end{align*}
Thus, we can pick $\epsilon$ small and $n$ large enough to require that the right-hand side of the inequality above is of constant order (recall that $\|B(\lambda)\|$ does not depend on $n$). In contrast, the average $n^{-1}\sum\limits_{i=1}^ny_i$ is of constant order if and only if $\|y\|_{\ell_1}=\Omega(n)$.
\end{proof}

As the primal-dual methods of~\cite{Brandao2017b,Apeldoorn2017,Apeldoorn2018} have a superquadratic dependency on $\|y\|_{\ell_1}$ for approximately optimal solutions, we conclude that their performance is worse than our algorithm for instances of \textsc{MaxQP SDP} with $B(\lambda)$ for $\lambda>1$ with high probability. Of course, this only holds provided that the typical value of the SDP is a constant fraction away from the spectrum of $B(\lambda)$, as indicated by numerical evidence.

\section{Comparison to previous work and techniques for further improvement}\label{app:scalingerror}

This section is devoted to giving a brief overview over some promising proposals for speeding up SDP solvers for problems with a similar structure. The main message is that these unfortunately do not immediately apply to the general \textsc{MaxQP SDP} setting, especially for random signed matrices.

The main classical bottleneck behind Algorithm~\ref{alg:HU} is computing matrix exponentials. Dimension reduction techniques, like Johnson-Lindenstrauss, can sometimes considerably speed up this process, see e.g.~\cite{Arora2007}. There, Arora and Kale apply this idea to solve the \textsc{MaxCut SDP} up to a multiplicative error of $\cO(\epsilon nd)$ in time $\tcO(nd)$ for a $d$ regular graph on $n$ vertices. Moreover, sparsification techniques~\cite{1911.07306} can be used to bring this complexity down to $\tcO(n)$ in the adjacency list model and $\tcO(\min(nd,n^{1.5}d^{-1}))$ in the adjacency matrix input model. Note that the \textsc{MaxCut SDP} is just an instance of the \textsc{MaxQP SDP}, as both have the same constraints. The only difference is that the \textsc{MaxCut SDP} has the additional structure that the target matrix is the weighted adjacency matrix of a graph and, thus, has positive entries. %
The extra assumption of non-negative entries is a key ingredient behind the fastest approximate \textsc{MaxCUT SDP} solvers which would outperform the main results of this work. 
It is therefore worthwhile to discuss why these ideas do not readily extend to more general problem instances.

First, note that the fact that the entries of the target matrix has positive entries is crucial for the soundness of the oracle presented in~\cite[Theorem 5.2]{Arora2007}. This already rules out the possibility of directly applying their methods to \textsc{MAXQP} if the matrix $A$ has negative entries. 
The second crucial observation of~\cite{Arora2007} is that it is possible to rewrite the \textsc{MaxCut SDP} as:
\begin{align}\label{equ:maxcutvector}
\textrm{minimize} & \quad  \sum\limits_{i,j}\left[A\right]_{ij}\|v_i-v_j\|^2  \\
\textrm{subject to} & \quad \|v_i\|^2 =1, \;
 v_i\in\R^n,\; i \in \left[n\right] \nonumber
\end{align}
In this reformulation, the vectors $v_i$ correspond to columns of a Cholesky-decomposition associated with feasible points: $\left[X \right]_{ij} = \langle v_i| v_j \rangle$. 
Next, recall the following variation of the polarization identity:
\begin{align*}
\langle u|v \rangle=\frac{1}{2}\lb \|u\|^2+\|v\|^2-\|u-v\|^2\rb.
\end{align*}
This allows us to rewrite the original objective function as
\begin{align*}
\tr{AX}=\sum\limits_{i,j}\left[A \right]_{ij}\langle v_i|v_j \rangle=
\frac{1}{2}\sum\limits_{i,j}\left[A \right]_{ij}\lb \|v_i\|^2+\|v_j\|^2-\|v_i-v_j\|^2\rb.
\end{align*}
Feasibility of $X$ then demands $1=\langle i|X|i \rangle = \langle v_i|v_i \rangle = \| v_i \|^2$ and we, thus, only need to optimize over  $\|v_i-v_j\|^2$. Subsequently, Arora and Kale apply dimensionality reduction techniques to compute approximate vectors $v_i',v_j'$ that satisfy:
\begin{align}\label{equ:equationvectorsJL}
\left|\|v_i-v_j\|^2-\|v_i'-v_j'\|^2\right|\leq \epsilon \|v_i-v_j\|^2.
\end{align}
in time $\cO(ns)$. 
A priori, similar techniques can be applied to the more general \textsc{MaxQP SDP} \eqref{equ:maxqpsdp}. However, sign problems can substantially affect the approximation error. Pointwise estimates like the one in~\eqref{equ:equationvectorsJL} only suffice to estimate $\tr{XA}$ up to an error of order $\cO(\epsilon\|A\|_{\ell_1})$.
This is fine for matrices with non-negative entries, where this error scaling is comparable to the size of the optimal SDP solution. Matrix entries with different signs, however, may lead to cancellations that result in a much smaller size of the optimal SDP solution. 
In summary: adapting the ideas of Arora and Kale \cite{Arora2007} is advisable in situations where the problem matrix obeys $\| A \|_{\ell_1} = \Theta (n \|A \|)$. This ensures a correct error behavior and dimension reduction allows for reducing the classical runtime to  
 $\tcO(ns)$. 

Another important technique for complexity reduction in SDPs is sparsification. Once again, one seminal example is \textsc{MaxCut}, where spectral sparsification methods can be used to reduce the complexity~\cite{Spielman2011,Kyng2016}. Here, the idea is to find a (usually random) sparser matrix $B$ that has approximately the same cut value as $A$ and then run the algorithm on $B$ instead. Unfortunately, once again signed matrix entries render this approach problematic. Up to our knowledge, the best current sparsification results available for the $\infty\to1$ norm are those of~\cite[Chapter 3]{Gittens2013}. There, the author shows in Corollary 3.9 that if we let $B$ be a random matrix with independent random entries s.t. $\mathbb{E}(B_{ij})=A_{ij}$, then
\begin{align*}
\mathbb{E}\left[ \|A-B\|_{\infty\to 1}\right] \leq2 \sum_{i}\sqrt{ \sum_j \operatorname{Var}[ B_{ij}]}.
\end{align*}
A necessary pre-requisite for accurate sparsification using the aforementioned result is therefore 
\begin{align*}
2 \sum_{i}\sqrt{ \sum_j \operatorname{Var}[ B_{ij}]}=\cO (\epsilon n \|A \|)
\end{align*}
It seems unlikely that it is possible to obtain good and general sparsification bounds from this result in our setting. To see why this is the case, note that in order for $B$ to be sparse in expectation, we require that $\mathbb{P}(B_{ij}=0)=p_{ij}$ for suitably large $p_{ij}$. This will result in a matrix that has, in expectation, 
$
\sum_{ij}(1-p_{i,j})
$
nonzero entries. 
To make sure that the number of nonzero entries is not $\cO(n^2)$, we need to set many $1-p_{ij}=o(1)$.
Now note that $\mathbb{P}[B_{ij}=0]=p_{ij}$ and $\mathbb{E}[B_{ij}]=A_{ij}$ necessarily enforce  $\mathbb{E}[(B_{ij}^2)]\geq \tfrac{p_{ij}}{1-p_{ij}}A^2_{ij}$.
Thus, we see that we expect this technique to only work in the regime where $A$ has many columns with entries that are $o(1)$ and can be neglected with high probability. Roughly speaking, this corresponds to the regime in which
$
\|A\|_{\textrm{col}}\ll n\|A\|.
$
It is then easy to see that the random matrices considered before do not satisfy this and, thus, we do not expect that those instances can be sparsified.

Last, we emphasize that it easy to construct examples where the error term $\| A \|_{\ell_1}$ conveys the right scaling, not $n \| A \|$. A concrete example are extremely sparse matrices, where all but $s \ll n$ of the entries are zero.

\end{document}